%% file: main.tex
\documentclass[sigconf]{acmart}

\renewcommand\footnotetextcopyrightpermission[1]{}
\usepackage{amsmath,amssymb,amsthm,mathtools}

\acmConference{}{}{}
\acmBooktitle{}
\acmPrice{}
\acmDOI{}
\acmISBN{}

\usepackage{graphicx}
\usepackage{subfigure}
\usepackage{textcomp}

\usepackage{booktabs}
\usepackage{multirow}
\usepackage{makecell}
\usepackage{threeparttable}

\usepackage{enumitem}

\usepackage{pgfplots}
\pgfplotsset{compat=1.18}

\usepackage{algorithm}
\usepackage{algpseudocode}

\usepackage[nameinlink,noabbrev]{cleveref}

\setlist[itemize]{leftmargin=*,topsep=2pt,itemsep=1pt}
\setlist[enumerate]{leftmargin=*}

\newtheorem{problem}{Problem}
\newtheorem{theorem}{Theorem}
\newtheorem{lemma}[theorem]{Lemma}

\newtheorem{corollary}[theorem]{Corollary}

\theoremstyle{definition}
\newtheorem{definition}[theorem]{Definition}

\begin{document}

\title{Finding the Balance Rate of Uncertain Signed Graphs}

\author{Zeyu Wang}
\affiliation{
  \institution{Zhejiang University}
  \country{}
}

\author{Kudria Sergei}
\affiliation{
  \institution{The Chinese University of Hong Kong, Shenzhen}
  \country{}
}

\author{Jingbang Chen}

\affiliation{
  \institution{The Chinese University of Hong Kong, Shenzhen}
  \country{}
}

\author{Jiawei Chen}
\affiliation{
  \institution{Zhejiang University}
  \country{}
}

\author{Xinyu Wang}
\affiliation{
  \institution{Zhejiang University}
  \country{}
}

\author{Xiaodong Luo}
\affiliation{
  \institution{Shenzhen Research Institute of Big Data}
  \country{}
}

\author{Can Wang}
\affiliation{
  \institution{Zhejiang University}
  \country{}
}

\renewcommand{\shortauthors}{Wang et al.}

\begin{abstract}
Signed graphs are widely used to analyze complex systems such as social, political, and biological networks. The notion of balance, a key concept of signed graphs, reflects the stability of relationships. While it has been extensively studied in deterministic graphs, real-world networks often exhibit uncertainty in their connections, which traditional approaches struggle to address. To bridge this gap, we introduce the concept of balance rate, a metric for quantifying the degree of balance in uncertain signed graphs, and prove that computing it exactly is NP-hard, motivating the need for efficient estimation methods. We propose a novel Rao–Blackwellized spanning-tree estimator that achieves near-linear time complexity per sample by leveraging graph decomposition and structural properties. We also construct asymptotically justified confidence intervals using the Delta method. Experiments on real-world datasets demonstrate the efficiency and effectiveness of our approach, enabling scalable balance analysis in uncertain signed graphs.
\end{abstract}

\maketitle

\input{intro}
\input{relatedwork}
\input{overview}
\input{algorithm}
\input{bounds}
\input{experiments}
\input{conclusion}

\clearpage
\bibliographystyle{ACM-Reference-Format}
\bibliography{ref}

\clearpage
\appendix
\input{appendix}
\end{document}

%% file: intro.tex
\section{INTRODUCTION}

Signed graphs, where edges represent positive or negative relationships, provide a powerful framework for analyzing interactions in complex systems \cite{harary1953notion}. As a fundamental concept in signed graphs, balance theory reflects the stability of relationships of the network \cite{cartwright1956,harary1959measurement}. Generally speaking, a graph is balanced if it can be divided into two groups where all edges within groups are positive while all edges between groups are negative. Such a notion indicates the structural properties of real-world networks. For instance, studies on social media platforms, such as Reddit or Twitter, have shown that polarized groups often emerge, where users within the same group exhibit positive interactions (e.g., likes), with negative interactions (e.g., conflicts) between groups \cite{doreian1996partitioning,yang2007community,giscard2017evaluating,facchetti2011computing,xiong2020logic}. In political networks, balance theory has been applied to analyze party alliances and conflicts, revealing the relationships between political factions \cite{adamic2005political,aref2020detecting,huang2022pole,capozzi2023analyzing}. In bioinformatics, the stimulatory and inhibitory interactions of genes represent positive and negative edges respectively, and a balanced network identifies functionally inhibitory gene groups \cite{peter2019encyclopedia,yang2020graph,liu2018multi,zhang2012signed}. It is also commonly applied in ecological \cite{saiz2018,hashemi2022review,xiao2017mapping} and financial networks \cite{zhang2017,ehsani2020structure,boginski2005statistical}.

However, the concept of balance is strict: the entire system's balance can be disrupted by few edges. Therefore, in reality, a perfectly balanced graph does not usually appear, preventing us from capturing the actual community. To address this, one possible solution is to study the optimal balance in unbalanced graphs, such as extracting a maximum balanced subgraph \cite{figueiredo2014maximum,ordozgoiti2020finding,chen2024scalable} or computing frustration index \cite{harary1959measurement,aref2020modeling}. Another way is to account for edge uncertainty. In practice, relationships may fade, or re-emerge due to changes in external conditions, shifts in shared interests, or the resolution of conflicts over time. Such uncertainty of relationships reflects the inherent dynamics of real-world systems. We consider this dynamic nature as edge existence probability, and model such systems as uncertain signed graphs (USGs). Then a natural question arises: \emph{what is the probability that the network is balanced?} We term this extension of balance as the \emph{balance rate}. It enhances the applicability of signed graph theory in uncertain environments. Specifically, it reveals diffusion efficiency in social platforms \cite{zhang2017}, polarization evaluation in political networks \cite{doreian2019}, and gene anomalies in bioinformatics \cite{costanzo2016, liu2015}.

This paper investigates the balance rate in Uncertain Signed Graphs (USGs), a metric that quantifies structural stability under edge uncertainty. We formally define the balance rate, establish its fundamental properties, and prove that its exact computation is NP-hard. To address this computational challenge, we propose an efficient approximation framework centered on a structural decomposition that partitions the USG into biconnected components via the identification of bridges and articulation points. This preprocessing step reduces the global problem into edge-independent sub-problems, significantly mitigating sampling variance. We further enhance estimation precision through a Rao-Blackwellized Monte Carlo approach; by leveraging the algebraic properties of the graph's spanning trees, we achieve substantial variance reduction over standard Monte Carlo methods, ensuring both high precision and computational tractability. The framework is supported by rigorous theoretical confidence bounds. Finally, comprehensive experiments on synthetic and real-world datasets demonstrate the scalability of our algorithm, its ability to capture intuitive structural trends, and its practical utility in identifying critical edges.

The contributions of this paper are summarized as follows:
\begin{itemize}[leftmargin=0pt]
    \item We propose the concept of balance rate as a novel yet important metric for balance evaluation in uncertain signed graphs, and discuss its NP-hardness mathematically.
    \item We develop an efficient and unbiased balance rate estimator, based on Rao-Blackwellized spanning-tree, that leverages graph decomposition and structural properties to reduce variance, achieving nearly linear time complexity for each sample.
    \item We construct non-parametric confidence intervals for global balance rate, using the Delta method to account for the specific variance and skewness of the Rao-Blackwellized samples, ensuring robust and adaptive inference.
    \item Evaluating our proposed algorithm on synthetic and real-world datasets, we first establish its computational scalability and robustness across diverse network sizes and densities. We then demonstrate that the metric faithfully captures the intuitive trend of balance through edge-probability modulation and injection of critical edges, ultimately showcasing its practical utility for the efficient identification of critical edges that drive network imbalance.\footnote{Codes are available at https://github.com/SiggyLaewsky/Finding-the-Balance-Rate-of-Uncertain-Signed-Graphs}
\end{itemize}

\noindent\textit{Outline.} The remainder of this paper is structured as follows.
Section~\ref{sec:related_work} reviews existing methodologies for signed graph analysis. Section~\ref{sec:definitions} formalizes the balance rate for Uncertain Signed Graphs (USG) and establishes its fundamental properties, including a proof of its NP-completeness. We detail our algorithmic framework for balance rate evaluation in Section~\ref{sec:algo} and derive its theoretical PAC bounds. Section~\ref{sec:experiments} provides a comprehensive empirical evaluation using synthetic and real-world datasets, focusing on computational scalability, the intuitive trend of balance, and the identification of critical edges. Finally, we offer concluding remarks and summarize our findings in Section~\ref{sec:conclusion}.

%% file: relatedwork.tex
\section{RELATED WORK}\label{sec:related_work}
\noindent\textit{Uncertain Graphs.}
Researches on uncertain graphs are thriving in a wide range of fields. A classic problem is the reliability query \cite{harary1967graph,aggarwal1975reliability}. It computes the probability that two nodes are connected. Generalizations like conditional reliability \cite{khan_conditional_2018}, expected reliable distance \cite{potamias_fast_2009} and distance-constrained reliability \cite{jin_distance-constraint_2011} also serve as basic metrics reflecting the connectivity of nodes, showing potential in fields like social and bioinformatics networks \cite{li_discovering_2017,jin2010computing}.
Most uncertain graph query problems are proven to be \#P-Hard \cite{ball_computational_1986,valiant1979complexity}. Therefore, they are usually solved by Naive Monte-Carlo (NMC). It is easy to implement, and the result's accuracy can be guaranteed by the sampling size $K$ \cite{fishman_comparison_1986}. However, the scalability is limited as graph scale increases. To address this issue, some researchers design heuristic sampling skills to improve efficiency, such as recursive sampling \cite{jin_distance-constraint_2011}, lazy propagation sampling \cite{li_discovering_2017} and recursive stratified sampling \cite{li_recursive_2016}. Other researchers apply indexing structures to answer queries by labeling index directly, like fixed width decomposition \cite{maniu_indexing_2017} and UR-index \cite{wang2024fast}.

\vspace{0.25cm}
\noindent\textit{Graph Community Mining.}
Community detection is another attractive task in graph data mining \cite{fortunato2010community}. It can be regarded as graph clustering problem, including $k$-core and $k$-truss decomposition \cite{bonchi2014core,huang2016truss}, induced subgraph mining \cite{zou2013polynomial}, subgraph similarity search \cite{yuan2012efficient}, dense subgraph discovery\cite{zou2013polynomial} and frequent subgraph patterns discovery \cite{zou2013polynomial,papapetrou2011efficient}. Such problems help find stable and robust structures in graphs. There are also variant clustering problems on special graphs, such as bipartite graphs \cite{zha2001bipartite}, uncertain graphs \cite{yang2019index,wen2020computing,zou2017truss} and so on.

\vspace{0.25cm}
\noindent\textit{Signed Graphs.}
Signed graphs were first introduced by Harary et al. \cite{harary1953notion}. They also proposed the concept of balanced graph, and explored methods to determine balance. Then they generalized Heider’s theory of balance onto signed graphs \cite{cartwright1956}, developed an algorithm to detect balance in signed graphs \cite{harary1980simple}. Since graphs are usually not balanced, some researches aim to find ways to make the graph balanced. Frustration index was introduced to count the minimal number of edges to be removed to make the whole graph balanced \cite{harary1959measurement}, and Aref et al. provided exact computation for it \cite{aref2020modeling}. These researches aim to break imbalanced structures at least cost, while others focus on finding maximum balanced subgraphs \cite{figueiredo2014maximum,ordozgoiti2020finding}. And Chen et al. proposed a scalable algorithm for finding maximum balanced graph with tolerance to meet real scenario standards \cite{chen2024scalable}. 
With edges assigned to existence probabilities in real-world scenarios, interests in uncertain signed graphs are increasing. Mandaglio et al. first proposed the uncertain signed graph clustering problem, and designed efficient partition and heuristics \cite{mandaglio2020and}. Then Zhang et al. introduce the concept of antagonistic communities to discover opponent groups, and leverage a local search method to identify them \cite{zhang2024finding}. But the most fundamental property, balance of an uncertain signed graph, remains underexplored, and requires a comprehensive theoretical framework.

%% file: overview.tex
\section{PROBLEM OVERVIEW}\label{sec:definitions}

\subsection{Preliminaries and Definitions}

Our research focuses on the balance theory on uncertain signed graphs. An \emph{uncertain signed graph} $G=(V,E,P,S)$ is composed of a vertex (node) set $V$ and an edge set $E$. Each edge $e\in E$ exists independently with probability $p_e\in [0,1]$ and a fixed sign $s_e\in S=\{+,-\}$. The edge signs are deterministic, where $E^+$ denotes positive edges with $s_e=``+"$ and $E^-$ denotes negative edges, while uncertainty is restricted to edge existence. If we decide the existence of edge by $p_e$, we get a certain subgraph $g=(V_g,E_g,S_g)$ of $G$. Such a $g$ is called a realization, or possible world, and the probability of $g$ is given by
\[
\Pr(g)=\prod_{e\in E_g} p_e \prod_{e\notin E_g} (1-p_e).
\]

The concept of balance is derived from certain signed graphs. We first give the formal definition of balanced graphs as follows, which is the same as previous works \cite{harary1953notion,cartwright1956}. Note that we and these works require the graph to be connected for the community detection proposal.

\begin{definition}[Balanced Graph]
Given a signed graph $G=(V,E,S)$, $G$ is balanced if $G$ is connected and there exists a partition $V=V_1\cup V_2$, $V_1\cap V_2=\emptyset$ such that for each edge $(u,v)\in E^+$, $u,v$ belong to the same partition while for each edge $(u,v)\in E^-$, they belong to different partitions.
\end{definition}

Graphs are usually not balanced considering the strict conditions. Fortunately, in uncertain signed graphs, there always exists some realizations that are balanced. In this paper, we introduce balance rate as the probability of an uncertain graph being balanced. Here, we give the definition of balance rate formally:

\begin{definition}[Balance Rate ($R_{bal}$)]
Let $\mathcal{G}$ denote the set of all realizations of a USG $G$. The \emph{balance rate} of $G$ is defined as
\[
R_{bal}(G)=\sum_{g\in\mathcal{G}}\Pr(g)\cdot\mathbb{I}(g\text{ is balanced}),
\]
where $\mathbb{I}$ is a binary indicator function. We sometimes abuse the definition of $R_{bal}$ to denote its expectation as well, which is clear from the context.
\end{definition}

Balance rate quantifies the stability of the network, providing a fundamental measurement for USG analysis. To further investigate this property, we formally define the computational problem as follows:

\begin{problem}[$R_{bal}$ computation]
Given an uncertain graph $G=(V,E,P,S)$, find the balance rate $R_{bal}(G)$ of $G$.
\end{problem}

\subsection{Basic Properties of Balance Rate}

\setlength{\abovedisplayskip}{0.5pt}
\setlength{\belowdisplayskip}{0.5pt}
The definition of balanced graphs reveals the structural property of signed graphs. Several methods have been proposed to identify a balanced subgraph, among them, a practical one is based on negative cycle detection \cite{harary1959measurement}. Here, a negative cycle is the cycle whose number of negative edges is odd. 

\begin{theorem}\label{thm:balance_detection}
A signed graph $G$ is balanced, if and only if there is no negative cycle in $G$.
\end{theorem}

To exploit the structural sparsity of the network, we decompose the USG into biconnected components. We define a bridge as an edge whose removal increases the number of connected components, and an articulation point as a vertex whose removal similarly disconnects the graph. The following two lemmas establish that the balance rate is multiplicative across these partitions, allowing us to factorize the global probability into independent subproblems.

\begin{lemma}[Bridge decomposition]\label{lem:bridge_decomp}
Let $e=(u,v)$ be a bridge of $G$, and let $G_1$ and $G_2$ be the two connected components of $G$ decomposed by $e$, with $V(G_1)\cap V(G_2)=\emptyset$ and $V(G_1)\cup V(G_2)=V$, where $u\in V(G_1)$ and $v\in V(G_2)$. Then $R_{bal}(G)=R_{bal}(G_1)\cdot R_{bal}(G_2).$
\end{lemma}
\begin{lemma}[Articulation point decomposition]\label{lem:cutpoint_decomp}
Let $v$ be an articulation point of $G$, and let $G_1, \ldots, G_k$ be the  connected components of $G$ decomposed by $v$, with $V(G_i)\cap V(G_j)=\{v\}$ for any $1 \leq i, j \leq k$ and $V(G_1)\cup \ldots \cup V(G_k)=V$. Then $R_{bal}(G)=\prod \limits_{i = 1}^k R_{bal}(G_i).$
\end{lemma}
The proofs of Lemmas~\ref{lem:bridge_decomp} and~\ref{lem:cutpoint_decomp} are presented in \Cref{app:missing_proofs}.

Lemmas~\ref{lem:bridge_decomp} and~\ref{lem:cutpoint_decomp} imply that, by decomposing the graph from bridges and articulation points, $G$ is separated into highly connected subgraphs, while remaining the same balance rate as the original graph. Therefore, the balance rate can be calculated over blocks. 

\begin{corollary}[Blockwise factorization]
Let $G_1,\ldots,G_k$ be the maximal connected subgraphs (blocks) decomposed by removing all bridges and articulation points of $G$, then $R_{bal}(G)=\prod_{j=1}^k R_{bal}(G_j).$
\end{corollary}

\subsection{Hardness Analysis}

In this subsection, we establish the NP-hardness of computing the balance rate. 

Frustration index is the minimum number of edges whose deletion makes all connected components balanced, denoted as $L(G)$ \cite{abelson2017symbolic,harary1959measurement}. The computation of frustration index can be reduced to the classical \textsc{EdgeBipartization} problem on unsigned graphs\cite{aref2020modeling}, which asks for a minimum-cardinality set of edges whose removal results in a bipartite graph \cite{edge_deletion}. Since \textsc{EdgeBipartization} is NP-hard \cite{hartmanis1982computers}, computing the frustration index is NP-hard as well.

In the following theorem, we provide an insight that balance rate computation is inherently difficult by exhibiting a polynomial-time reduction from the frustration index problem.

\begin{theorem}\label{thm:NP-hard}
  There is no polynomial time algorithm for precise computation of $R_{bal}$, unless $\mathbf{P}=\mathbf{NP}$.
\end{theorem}
\begin{proof}
  For a signed graph $G = (V,E,S) $, the frustration index $L(G)$ the minimum number of edges whose deletion makes \(G\) balanced. Deciding whether \( L(G) \le k \) is known to be NP-hard. We reduce this problem to thresholding \( R_{bal} \).
	Given a signed graph \( G = (V, E, S) \) we construct a probabilistic signed graph on the same vertices and edges, with the same signs, where each edge $e \in E$ exists independently with probability 
	\( p_e = 1 - \varepsilon \) for $\varepsilon = 2^{-(|E|+2)}$. This choice of \( \varepsilon \) has polynomial encoding length in $|E|$, so $p_e$ does.
	If \( L(G) \le k \), then there exists a subset \( F \subseteq E \) with 
	\(|F| \le k\) such that \(E \setminus F\) induces a balanced subgraph. 
	The probability that exactly the edges of \(F\) are missing and all others are present is
	\[
	\Pr(E^\prime = E \setminus F) = \varepsilon^{|F|}(1-\varepsilon)^{|E| - |F|} \ge \varepsilon^{k}(1-\varepsilon)^{|E|-k}
	\]
	and therefore $R_{bal} \geq \varepsilon^k (1 - \varepsilon)^{|E|-k}$. Conversely, if \( L(G) \ge k + 1 \), then every balanced subgraph must result 
	from deleting at least \( k+1 \) edges, and thus
	\begin{align*}
		\begin{split}
			R_{bal} &= \sum_{\substack{M \subseteq E \\ (V,E\setminus M,S) \;balanced}}
			\varepsilon^{|M|}(1-\varepsilon)^{|E| - |M|} \leq \\
			&\leq \sum_{i=k+1}^{|E|} \binom{|E|}{i}\varepsilon^{i}(1-\varepsilon)^{|E|-i} \leq 2^{|E|}\varepsilon^{k+1}(1-\varepsilon)^{|E|-k-1}.
		\end{split}
	\end{align*}
	
	Taking the ratio of the lower and upper bounds gives
	\[
	\frac{\text{(LB)}}{\text{(UB)}} 
	\ge \frac{1-\varepsilon}{2^{|E|}\varepsilon}
	= (1-\varepsilon)\cdot 4 > 1,
	\]
	for the chosen \(\varepsilon = 2^{-(|E|+2)}\). 
	Thus the two cases \( L(G) \le k \) and \( L(G) \ge k+1 \) yield disjoint intervals 
	for possible values of \( R_{bal} \). 
	Setting $T = \frac{1}{2}\!\left[
	\varepsilon^{k}(1-\varepsilon)^{|E|-k}
	+ 2^{|E|}\varepsilon^{k+1}(1-\varepsilon)^{|E|-k-1}
	\right],$
	we have \( L(G) \le k \Rightarrow R_{bal} > T \) and 
	\( L(G) \ge k+1 \Rightarrow R_{bal} < T \). 
	Hence deciding whether \( R_{bal} > T \) is equivalent to deciding whether \( L(G) \le k \). 
	Since the latter problem is NP-hard, the computation (or even thresholding) of 
	\( R_{bal} \) is NP-hard as well.
\end{proof}

The NP-hardness of computing $R_{bal}$ highlights the impracticality of exact computation for large-scale graphs, demonstrating the necessity of developing efficient approximation algorithms.

%% file: algorithm.tex
\section{ALGORITHM}\label{sec:algo}

In this section we describe our sampling-based algorithmic framework for evaluating the balance rate, incorporating variance reduction techniques. Specifically, given a USG $G=(V,E,P,S)$, we first apply graph preprocessing by decomposing $G$ into a series of connected components (Section \ref{sec:preprocess}), then introduce CCI, a sampling process based on Rao-Blackwellized spanning tree in Section \ref{sec:rb_sample}
, and show its theoretical justification. The overall framework is provided in Section \ref{sec:sampling_all}, and we derive the confidence bounds by Delta Method in \Cref{sec:delta_method_bounds}.

\subsection{Graph Preprocessing}\label{sec:preprocess}

Intuitively, balance rate can be directly estimated by Naive Monte-Carlo (NMC) algorithm. However, NMC treats the graph as a monolithic entity, where the global balance outcome is a binary variable. As a result, it fails to exploit the independence of graph components, leading to slow convergence. From Lemmas \ref{lem:bridge_decomp} and \ref{lem:cutpoint_decomp}, topological features such as bridges and articulation points divide the graph into independent components. By removing them, we replace the single high-variance global $R_{bal}$ estimation with several lower-variance local balance rates, ensuring that noise from one dense component does not affect other parts. Inspired by this, we apply a preprocessing step to decompose the graph into highly connected subgraphs, enabling faster convergence and stable variance by handling components independently, within which connectivity is more homogeneous.

We formalize the preprocessing procedure in \Cref{alg:bridge_cutpoint_decomposition}. For a given USG, we first identify all bridges and articulation points. The identification is based on Tarjan's algorithm \cite{tarjan1972depth}, conducted in a depth-first search (DFS) traversal (Line \ref{line:init_start}-\ref{line:init_end}). While searching, we continuously apply Tarjan's algorithm to probe whether the currently visited edge or vertex meets the conditions. If so, we use two sets, $B$ and $A$ to store bridges and articulation points separately. After DFS, we remove all edges and vertices in $B$ and $A$ from $G$, then compute the connected components in the remaining graph $G'$. After \Cref{alg:bridge_cutpoint_decomposition}, the USG is decomposed into connected components $\{C_1,C_2,\dots,C_k\}$, where $E(C_i\cap C_j)=\emptyset$ and $\bigcup_{i=1}^k V(C_i)=V$ for $1\leq i,j\leq k, i\neq j$. This decomposition is purely structural and does not alter any statistics computed within the components; it only reduces interdependence and variance across them.

\begin{algorithm}[t]
    \caption{Graph Decomposition}
    \label{alg:bridge_cutpoint_decomposition}
    \begin{algorithmic}[1]
        \Require A USG $G=(V,E,P,S)$
        \Ensure A set of connected components $\mathcal{C}$ of $G$
        \State Initialize $B \leftarrow \emptyset$ (bridges), $A \leftarrow \emptyset$ (articulation points)
        \ForAll{$v\in V$}
        \State Initialize $low[v]\gets0,disc[v]\gets0$
        \EndFor
        \State Run DFS on $G$ and update $low[v],disc[v]$ while visiting\label{line:init_start}
        \ForAll{edges $(u,v) \in E$ explored by DFS}
        \If{$low(v) > disc(u)$} \Comment{Meet bridge conditions}
        \State Add $(u,v)$ to $B$
        \EndIf
        \EndFor
        \ForAll{vertices $v \in V$ visited by DFS}
        \If{$v$ meets articulation condition by Tarjan's}
        \State Add $v$ to $A$
        \EndIf
        \EndFor\label{line:init_end}
        \State $G' \leftarrow G$ with all edges in $B$ and all vertices in $A$ removed
        \State $C_1,C_2,\dots,C_k\gets$ connected components of $G'$
        \State $\mathcal{C}\gets \{C_1,C_2,\dots,C_k\}$
        \State \Return $\mathcal{C}$
    \end{algorithmic}
\end{algorithm}

\vspace{0.25cm}
\noindent \textit{Complexity analysis.}
\Cref{alg:bridge_cutpoint_decomposition} runs a single DFS to identify all bridges and articulation points. This requires $O(|V|+|E|)$ time and space. Removing the identified edges and vertices and computing the connected components of the remaining graph also takes
$O(|V|+|E|)$ time. Hence, the total time complexity of the preprocessing step is linear in
the size of the input graph.

The memory usage is also linear, dominated by the storage of DFS visiting variables (discovery times, low-link values, and recursion stack). Since the procedure is executed only once as a preprocessing step, its overhead is negligible compared to the sampling algorithms
applied independently to each resulting component.

\vspace{0.25cm}
\noindent \textit{Variance analysis.}
We next prove that, for any estimator whose construction is based on our preprocessing, sampling on the decomposed components strictly dominates sampling on the original graph in terms of variance, unless in the degenerate case where all components are deterministic.

\paragraph{Estimator A (joint sampling).}
Draw $N$ independent realizations $\left(\hat{R}_{bal,1}^{(k)},\dots,\hat{R}_{bal,m}^{(k)}\right)_{k = 1}^N,$
and define
\begin{equation}
    \widehat R_{joint}
    \;=\;
    \frac{1}{N}\sum_{k=1}^N
    \hat{R}_{bal}^{(k)},
    \qquad
    \hat{R}_{bal}^{(k)} := \prod_{j=1}^m \hat{R}_{bal,j}^{(k)}.
\end{equation}

Estimator A corresponds to sampling on the original graph. 

\paragraph{Estimator B (product of marginal means).}
For each $j\in\{1,\dots,m\}$, draw $N$ independent samples
$\{\hat{R}_{bal,j}^{(k)}\}_{k=1}^N$ and define
\[
 \bar R_{bal,j}
:=
\frac{1}{N}\sum_{k=1}^N \hat{R}_{bal,j}^{(k)},
\qquad
\widehat R_{\mathrm{prod}}
:=
\prod_{j=1}^m \bar{R}_{bal,j}.
\]

Estimator B corresponds to sampling on the decomposed graph. 

Independence of the $R_{bal,j}$ immediately implies
$\mathbb{E}[\widehat R_{\mathrm{joint}}]=\mathbb{E}[\widehat R_{\mathrm{prod}}]=\mathbb{E}[R_{bal}].$

The following theorem shows that $\widehat R_{\mathrm{prod}}$ has uniformly smaller variance. We defer the proof to \Cref{app:missing_proofs}.

\begin{theorem}\label{thm:decomp_var_reduction}
    \label{thm:variance-reduction}
    For all $N\ge 1$, $\mathrm{Var}(\widehat R_{prod})\le \mathrm{Var}(\widehat R_{joint}),$
    with strict inequality unless all $\hat{R}_{bal,j}$ are almost surely constant.
\end{theorem}

\subsection{Cycle-Chord Integration Sampling}\label{sec:rb_sample}

In USGs, the event of structural balance is often rare. However, standard sampling instantiates every edge, introducing unnecessary Bernoulli noise for edges irrelavant to balance constraints, thus resulting in a high-variance estimator. To address this, we propose Cycle-Chord Integration (CCI) sampling, to obtain a spanning tree based estimator, guaranteeing a strictly lower variance than standard sampling.

The key point of CCI is to develop type-specific sampling from a random spanning forest based on the classification of edges. From the balance condition, forests are clearly balanced, thus they can be regarded as a maximum balanced graph in a connected component, enabling a non-zero balance rate. We partition the edge set into \textit{tree edges}, which we sample to determine connectivity, and \textit{chord edges}, which we treat analytically. Specifically, we only sample the tree edges to connect disjoint trees for a basic balanced sub-realization. Then for chord edges, we compute the probability that the cycle remains balanced, updating our estimator weight. By analytically integrating out the randomness of these cycle-closing edges, we replace binary outcomes with continuous expectations, yielding a lower-variance estimator.

We use an extension of the classical Disjoint Set Union (DSU) for efficient data management. To elaborate, we augment each vertex $v$ with a $\mathbb{Z}_2$-valued label $\rho_v$, representing the parity of negative edges on the path to its component root. This augmentation allows DSU to enforce the invariant $\rho_u\oplus\rho_v=s$ (XOR computation) for any signed edge $(u, v)$ with sign $s\in\{0, 1\}$, ensuring consistent parity within connected components, further detecting negative cycles directly. Moreover, it preserves the amortized $O(\alpha(n))$ complexity of standard DSU \cite{galil1991data}, as parity updates during path compression incur only constant-time overhead. 

The detailed procedure of CCI is presented in \Cref{alg:rb_sample}. The algorithm is implemented on the connected components obtained from \Cref{alg:bridge_cutpoint_decomposition}. We begin by initializing the DSU structure with parity tracking, and set the initial balance rate as 1. As for sampling stage, the algorithm distinguishes edges between two cases: tree edges and chord edges. For a tree edge, we use a Bernoulli trial to determine whether it is included in the spanning forest, and update the DSU structure accordingly. Otherwise, we detect whether it causes negative cycles by evaluating the parity of the current path and comparing it with the edge's sign (Line \ref{line:NC_detect_start}-\ref{line:NC_detect_end}). If so, we adjust the $R_{bal}$ estimator by multiplying with edge absent probability in Line \ref{line:br_update}. \Cref{alg:rb_sample} ultimately returns an unbiased estimator for the balance rate, achieving improved efficiency and accuracy in the estimation process.

\begin{algorithm}[t]
    \caption{Cycle-Chord Integration Sampling}
    \label{alg:rb_sample}
    \begin{algorithmic}[1]
        \Require A connected USG $G=(V,E,P,S)$
        \Ensure An unbiased estimator $\hat{R}_{bal}$
        \State Initialize DSU structure $\mathcal{D}$ with parity tracking $\rho$
        \State $\hat{R}_{bal} \leftarrow 1.0$
        \For{each edge $e=(u, v)$ with sign $s_e$ and prob $p_e$ in $E$}
        \State $(root_u, \rho_u) \leftarrow \mathcal{D}.\text{find}(u)$
        \State $(root_v, \rho_v) \leftarrow \mathcal{D}.\text{find}(v)$
        \If{$root_u \neq root_v$}  \Comment{Case 1: Tree Edge}
        \State Sample Bernoulli $X_e \sim \text{Ber}(p_e)$
        \If{$X_e = 1$}
        \State $\mathcal{D}.\text{union}(u, v, s_e)$
        \EndIf
        \Else   \Comment{Case 2: Chord Edge (Analytical Step)}
        \State $\sigma_{\text{path}} \leftarrow (\rho_u + \rho_v) \pmod 2$\label{line:NC_detect_start}
        \State $s_{\text{edge}} \leftarrow (1 \text{ if } s_e = -1 \text{ else } 0)$
        \If{$\sigma_{\text{path}} \neq s_{\text{edge}}$}\label{line:NC_detect_end}
        \State $\hat{R}_{bal} \leftarrow \hat{R}_{bal} \times (1 - p_e)$\label{line:br_update}
        \EndIf
        \EndIf
        \EndFor
        \State \Return $\hat{R}_{bal}$
    \end{algorithmic}
\end{algorithm}

\vspace{0.25cm}
\noindent \textit{Complexity analysis.}
CCI sampling relies on DSU with path compression and rank heuristics. The sampling iterate through all $|E|$ edges exactly once, and each DSU operation takes $O(\alpha(|V|))$ amortized time, where $\alpha$ is the inverse Ackermann function. Therefore, the total time complexity per sample is $O(|E| \alpha(|V|))$, which is nearly linear in the number of edges. This matches the complexity of Kruskal's algorithm, and is computationally equivalent to a single naive MC run, yet provides significantly more information.

\vspace{0.25cm}
\noindent \textit{Unbiasedness analysis.}
We show the unbiasedness of CCI estimator $\hat{R}_{bal}$. Let $I(\mathbf{X})$ be the indicator variable that is $1$ if $\mathbf{X}$ is balanced and $0$ otherwise. We define a filtration based on the spanning forest construction. Let $E_T \subseteq E$ be the set of edges processed as tree edges and $E_C = E \setminus E_T$ be the set of chord edges. Let $\mathcal{F}_T$ be the $\sigma$-algebra generated by the outcomes of edges in $E_T$.

The estimator produced by Algorithm \ref{alg:rb_sample} is $\hat{R}_{bal} = \mathbb{E}[I(\mathbf{X}) \mid \mathcal{F}_T]$.
By the Law of Iterated Expectations:
$\mathbb{E}[\hat{R}_{bal}] = \mathbb{E}[ \mathbb{E}[I(\mathbf{X}) \mid \mathcal{F}_T] ] = \mathbb{E}[I(\mathbf{X})] = R_{bal}.
$
Thus, $\hat{R}_{bal}$ is an unbiased estimator of the true balance rate $R_{bal}$.

\vspace{0.25cm}
\noindent \textit{Variance Reduction.}
$\hat{R}_{bal}$ can be analyzed using the law of total variance (Rao-Blackwell theorem \cite{wasserman2004all}):
$
    \text{Var}(I(\mathbf{X})) = \text{Var}(\mathbb{E}[I(\mathbf{X}) \mid \mathcal{F}_T]) + \mathbb{E}[\text{Var}(I(\mathbf{X}) \mid \mathcal{F}_T)].
$
Rearranging for the variance of our estimator $\hat{R}_{bal} = \mathbb{E}[I(\mathbf{X}) \mid \mathcal{F}_T]$,
$
    \text{Var}(\hat{\pi}) = \text{Var}(I(\mathbf{X})) - \mathbb{E}[\text{Var}(I(\mathbf{X}) \mid \mathcal{F}_T)].
$
Since variance is always non-negative, $\mathbb{E}[\text{Var}(I(\mathbf{X}) \mid \mathcal{F}_T)] \ge 0$, Therefore $\text{Var}(\hat{R}_{bal}) \le \text{Var}(I(\mathbf{X}))$.
The term $\mathbb{E}[\text{Var}(I(\mathbf{X}) \mid \mathcal{F}_T)]$ represents the uncertainty eliminated by analytically integrating the chord edges. In sparse graphs, where $|E_C|$ is small relative to $|E|$, the reduction is substantial because we avoid the binary noise associated with the most critical, conflict-generating edges.

\subsection{Overall Framework}\label{sec:sampling_all}

The overall balance rate estimation framework is summarized in \Cref{alg:sampling_all_}. It combines graph preprocessing (\Cref{alg:bridge_cutpoint_decomposition}) and Cycle-Chord Integration sampling (\Cref{alg:rb_sample}) to provide an efficient and unbiased estimator for the balance rate in USGs. By structural decomposition and Rao-Blackwellization, the framework significantly reduces variance compared to NMC methods. 

\begin{algorithm}[t]
    \caption{Overall Framework for Balance Rate Estimation}
    \label{alg:sampling_all_}
    \begin{algorithmic}[1]
        \Require A USG $G=(V,E,P,S)$, sample size $N$
        \Ensure An unbiased estimator $R_{bal}$ of $R_{bal}(G)$
        \State $\{G_1,\ldots,G_k\} \gets$ \Cref{alg:bridge_cutpoint_decomposition}$(G)$
        \For{$i = 1,\ldots,k$}
            \State $\hat{R}_{bal,i} = 0$
            \For{$j = 1, \ldots, N$}
                \State $\hat{R}_{bal,i}^{(j)} \gets$ \Cref{alg:rb_sample}$(G_i)$
                \State $\hat{R}_{bal,i} \gets \hat{R}_{bal,i} + \hat{R}_{bal,i}^{(j)} $
            \EndFor
            \State $\hat{R}_{bal,i} \gets \hat{R}_{bal,i} / N$
        \EndFor
        \State $\hat{R}_{bal}\gets\prod_{i=1}^k \hat{R}_{bal,i}$
        \State \Return $\hat{R}_{bal}$
    \end{algorithmic}
\end{algorithm}

Given a USG $G$ and sample size $N$, \Cref{alg:sampling_all_} begins by applying preprocessing to decompose the graph into independent connected components without bridges and articulation points. Then for each decomposed $G_i$, we conduct CCI sampling for $N$ times to get the local balance rate estimator $\hat{R}_{bal,i}$. The global balance rate $\hat{R}_{bal}$ is then computed as the product of all local estimates. It remains the unbiasedness and low variance obtained from \Cref{alg:rb_sample}, and the total complexity is 
$O(|V|+|E|+N\cdot\sum_{i=1}^k|E_i|\cdot\alpha(|V_i|)),$
which is nearly linear in the graph scale for a practical $N$. Note that $\sum_{i=1}^k|E_i|$ is the number of non-bridge edges, therefore, the algorithm is particularly efficient on sparse graphs.

\Cref{alg:sampling_all_} provides a scalable and robust solution for balance rate estimation, combining theoretical rigor with practical efficiency. It is particularly effective for large, sparse graphs, where structural decomposition significantly reduces complexity and variance. It is also well-suited for graphs with rare balance events, where naive sampling methods struggle to get effective results. 

Additionally, \Cref{alg:sampling_all_} is naturally parallelizable with respect to the Rao-Blackwellized Monte Carlo sampling loop. For each component $G_i$, the estimates $\hat{R}_{bal,i}^{(j)}$ are generated independently, allowing the $N$ sample evaluations of \Cref{alg:rb_sample} to be executed in parallel with no synchronization overhead. In terms of efficiency, the bridge/articulation point decomposition incurs a one-time $O(|E|)$ preprocessing cost. Since the per-sample complexity of \Cref{alg:rb_sample} is at least $O(|E|)$ and the number of samples is typically greater than $50$, the total runtime is dominated by sampling, making the preprocessing cost negligible in practice and resulting in excellent scalability.

%% file: bounds.tex
\subsection{Confidence Bounds}
\label{sec:delta_method_bounds}

In this section, we derive asymptotic confidence intervals for $R_{bal}$. Using the Multivariate Delta Method and Central Limit Theorem, we account for the Rao-Blackwellized nature and specific properties of the estimators. We establish a robust framework for computing confidence intervals that are both sensitive and asymptotically accurate for large sample sizes.

\vspace{0.25cm}
\noindent \textit{Balance Rate Estimator.}
In our analysis, we are interested in estimating the product of $\hat{R}_{bal,1}, \ldots, \hat{R}_{bal,k}$, where each $\hat{R}_{bal,j} \in [0, 1]$ represents a Rao-Blackwellized estimator. Since they are continuous with variance significantly lower than the Bernoulli variance $\pi(1-\pi)$, standard parametric likelihoods (e.g., Binomial) provide overly conservative and incorrect intervals. Let $\mu_j = \mathbb{E}[\hat{R}_{bal,j}]$ be the true mean of the $j$-th variable. By the independence of these estimators, our target parameter is $R_{bal} = \prod_{j=1}^k \mu_j.$
Given $N$ samples for each variable, we estimate $R_{bal}$ using the product of the empirical means $\hat{R}_{bal} = \prod_{j=1}^k \hat{\mu}_j$, where $\hat{\mu}_j = \frac{1}{N} \sum_{t=1}^N \hat{R}_{bal,j}^{(t)}.$

\setlength{\abovedisplayskip}{2pt}
\setlength{\belowdisplayskip}{1pt}

\vspace{0.25cm}
\noindent \textit{Asymptotic Justification by Delta Method}
We derive a $(1-\delta)$ confidence interval by the Multivariate Delta Method (see \cite{casella2002statistical}\cite{van2000asymptotic}). This method provides a functional form of the Central Limit Theorem (CLT). Specifically, let $\boldsymbol{\hat{\mu}} = (\hat{\mu}_1, \ldots, \hat{\mu}_k)^T$ and $g(\boldsymbol{\mu}) = \prod_{j=1}^k \mu_j$. The Delta Method states that as $N \to \infty$, we have $\sqrt{N}(\hat{R}_{bal} - R_{bal}) \xrightarrow{d} \mathcal{N}\left(0, \nabla g(\boldsymbol{\mu})^T \Sigma \nabla g(\boldsymbol{\mu})\right),$
where $\Sigma$ is the diagonal covariance matrix with entries $\Sigma_{jj} = \text{Var}(\hat{R}_{bal,j}) = \sigma_j^2$. 
The gradient $\nabla g(\boldsymbol{\mu})$ has components $\frac{\partial g}{\partial \mu_j} = \prod_{t \neq j} \mu_t$. The asymptotic variance $\sigma^2_{R}$ of our estimator $\hat{R}_{bal}$ is:
\begin{equation*}
  \sigma^2_{R} = \sum_{j=1}^k \left( \prod_{t \neq j} \mu_t \right)^2 \sigma_j^2 = \sum_{j=1}^k \left( \frac{R_{bal}}{\mu_j} \right)^2 \sigma_j^2.
\end{equation*}

In practice, the unbiased sample variance for each $\hat{R}_{bal,j}$ is:
\begin{equation*}
    \hat{\sigma}_j^2 = \frac{1}{N-1} \sum_{t=1}^N (\hat{R}_{bal,j}^{(t)} - \hat{\mu}_t)^2.
\end{equation*}

Then, we have estimated standard error $\widehat{SE}(\hat{R}_{bal})$ of the product:
\begin{equation*}
    \widehat{SE}(\hat{R}_{bal}) = \sqrt{\frac{1}{N} \sum_{j=1}^k \left( \prod_{t \neq j} \hat{\mu}_t \right)^2 \hat{\sigma}_j^2}.
\end{equation*}

The $(1-\delta)$ confidence interval is then given by
\begin{equation*}
    R_{bal} \in \left[ \hat{R}_{bal} - z_{1-\delta/2} \widehat{SE}(\hat{R}_{bal}), \quad \hat{R}_{bal} + z_{1-\delta/2} \widehat{SE}(\hat{R}_{bal}) \right].
\end{equation*}

For $N \geq 100$, the empirical means $\hat{\mu}_j$ are well-approximated by normal distribution via CLT, particularly for bounded variables in $[0, 1]$. In asymptotic statistics, $N \approx 30$ is often considered the Gaussian behavior threshold; at $N=100$, the first-order Taylor approximation utilized by the Delta Method is robust \cite{wasserman2004all}.

%% file: experiments.tex
\vspace{-0.2cm}
\section{EVALUATION} \label{sec:experiments}

In this section, we address the following research questions to
evaluate various important aspects of our algorithm:

\begin{itemize}[leftmargin=*]
    \item \textbf{RQ1 (Efficiency).} 
    How efficient and scalable is the proposed algorithm for computing the balance rate of uncertain signed graphs as the graph size, density, and uncertainty level increase?

    \item \textbf{RQ2 (Effectiveness).} Does balance rate provide an informative measure of balance in uncertain signed graphs?

    \item \textbf{RQ3 (Applicability).} Can the balance rate support practical graph analysis tasks beyond descriptive measurement?
\end{itemize}

\subsection{Experimental Settings}

\textit{Datasets.} We evaluate our algorithm in both real-world and synthetic USGs, with the main characteristics summarized in \Cref{tab:data} and \Cref{tab:synth_data} respectively. For real-world datasets, we list the category, vertex size $|V|$, edge size $|E|$, ratio of negative edges $\rho^-=\frac{|E^{-}|}{|E|}$ and ratio of non-zero elements $\delta=\frac{2|E|}{|V|(|V| - 1)}$ in \Cref{tab:data}. All real-world datasets are publicly-available
\footnote{snap.stanford.edu and konect.cc}.
Since real-world signed networks typically lack explicit edge existence probabilities, we generate synthetic probabilities to model structural uncertainty. To ensure our evaluation focuses on the non-trivial regime—avoiding saturation scenarios where the expected balance is asymptotically zero or one—we employ a density-adaptive scaling strategy. We sample mean edge probabilities from a uniform distribution scaled inversely by $|E| / \sqrt{|V|}$. This heuristic targets the phase transition of structural balance, ensuring the generated instances preserve complex cycle interactions without collapsing into trivial states of sparsity or deterministic rigidity.

For synthetic datasets, we construct from a randomly generated spanning tree on $n$ nodes, then add edges to $\lfloor 3n/2 \rfloor$ uniformly at random until the graph reaches a sparse regime. Finally, the network is densified by repeatedly closing length-two paths into triangles, until $|E|=5n$. The resulting topology combines a connected backbone, random long-range edges, and local triadic closures, yielding graphs with moderate density and nontrivial clustering. We endow the generated topology with existence probabilities following the heuristic described in the previous paragraph. We list $|V|$, $|E|$, $\rho^-$ and $\delta$ of synthetic networks in \Cref{tab:synth_data}.

\vspace{0.1cm}
\noindent \textit{Setting.} All algorithms were implemented in C++. Parallelization was performed using OpenMP. Unless stated otherwise, all experiments were executed using 4 OpenMP threads on an Apple M2 Pro processor with 12 physical cores (12 logical cores), running macOS 26.2. All reported runtimes correspond to this fixed parallel configuration on a shared-memory architecture. The code was compiled with compiler optimizations enabled, and identical settings were used across all experiments to ensure a fair comparison.

\subsection{Efficiency}


Table~\ref{tab:sampling-comparison} compares naive Monte Carlo (MC) sampling with the proposed Rao–Blackwellized (RB) estimator in terms of runtime and empirical variance across increasingly large sparse graphs. As graph size grows from tens to hundreds of thousands of vertices, both methods exhibit near-linear runtime scaling with respect to the number of edges, confirming the scalability of the sampling-based framework. While RB sampling incurs a modest constant-factor overhead in runtime relative to naive MC, this cost remains stable across all instances and does not grow disproportionately with graph size or the number of connected components.

Crucially, RB sampling achieves a dramatic and consistent reduction in estimator variance—by more than an order of magnitude in all cases—yielding substantially more stable estimates for the same number of samples. This variance reduction directly translates into improved robustness and faster convergence in practice, outweighing the additional per-sample cost. Overall, these results demonstrate that the proposed approach is computationally efficient at scale while providing significantly more reliable balance rate estimates than naive Monte Carlo sampling. We also provide the visualization  of prefix confidence intervals for both sampling strategies in \Cref{plot:conf_intervals}, computed as described in \Cref{sec:delta_method_bounds}.

\begin{table}
  \centering
    \caption{Real-world datasets: $|V|$, $|E|$, ratio of negative edges $\rho^{-}$ and ratio of non-zero elements $\delta$}
    \label{tab:data}
  \begin{tabular}{cccccc}
    \hline
    Dataset & Category & $|V|$ & $|E|$ & $\rho^-$ & $\delta$ \\
        \hline 
        \textsc{Bitcoin} & Trustness & 5k & 21k & 0.15 & 1.2e-03 \\
\textsc{Epinions} & Trustness & 131k & 711k & 0.17 & 8.2e-05 \\
\textsc{Slashdot} & Friendship & 82k & 500k & 0.23 & 1.4e-04 \\
\textsc{Twitter}  & Interaction & 10k & 251k & 0.05 & 4.2e-03 \\
\textsc{Conflict} & Conflicts & 116k & 2M & 0.62 & 2.9e-04 \\
\textsc{Elections}& Voting & 7k & 100k & 0.22 & 3.9e-03 \\
\textsc{Politics} & Interaction & 138k & 715k & 0.12 & 7.4e-05 \\ 
\hline
  \end{tabular}
\end{table}
\begin{table}

\vspace{-0.28cm}
	\centering 
    \caption{Synthetic datasets: $|V|$, $|E|$, ratio of negative edges $\rho^{-}$ and ratio of non-zero elements $\delta$}
	\label{tab:synth_data}
	\begin{tabular}{ccccc}
    \hline 
    Dataset & $|V|$ & $|E|$ & $\rho^-$ & $\delta$ \\
    \hline
    \textsc{sparse\_500} & 500 & 2.5k & 0.16 & 0.02\\
    \textsc{sparse\_5000} & 5k & 25k & 0.15 &2.0e-03\\
    \textsc{sparse\_10000} & 10k & 50k & 0.16 &1.0e-03\\
    \textsc{sparse\_50000} & 50k & 250k & 0.15 &2.0e-04\\ 
    \textsc{sparse\_100000} & 100k & 500k & 0.16 &1.0e-04\\
    \hline
    \end{tabular}
    \vspace{-0.3cm}
\end{table}

\begin{table*}[t]
\centering
\caption{Runtime, empirical variance, and balanced rate comparison between naive Monte Carlo (MC) sampling and Rao--Blackwellized (RB) sampling. Runtimes are reported for $200$ samples per dataset.} 
\label{tab:sampling-comparison}
\begin{tabular}{|c|cccccccc|}
\hline
Dataset & $|V|$ & $|E|$ & \#Comp &
MC time (ms) & MC var &
RB time (ms) & RB var & Balance rate \\
\hline
\textsc{sparse\_500}    & 0.5k   & 2.5k     & 2   & 1    & 1.2$\cdot 10^{-1}$ & 3    & 8.2$\cdot 10^{-3}$ & 0.143 \\
\textsc{sparse\_5000}     & 5k     & 25k      & 9   & 14   & 1.5$\cdot 10^{-1}$ & 24   & 3.7$\cdot 10^{-3}$ & 0.228 \\
\textsc{sparse\_10000}    & 10k    & 50k      & 9   & 26   & 1.8$\cdot 10^{-1}$ & 45   & 3.4$\cdot 10^{-3}$ & 0.245 \\
\textsc{sparse\_50000}    & 50k    & 250k     & 51  & 137  & 2.0$\cdot 10^{-1}$ & 207  & 1.3$\cdot 10^{-3}$ & 0.286 \\
\textsc{sparse\_100000}   & 100k   & 0.5M     & 111 & 268  & 2.0$\cdot 10^{-1}$ & 408  & 9.4$\cdot 10^{-4}$ & 0.289 \\
\textsc{bitcoin}        & 5.9k   & 21.5k    & 7   & 8    & 5.0$\cdot 10^{-3}$ & 24   & 1.4$\cdot 10^{-3}$ & 0.0072 \\
\textsc{elections}      & 7.1k   & 100.7k   & 1   & 55   & 2.3$\cdot 10^{-1}$ & 94   & 4.3$\cdot 10^{-2}$ & 0.403 \\
\textsc{twitter}        & 10.9k  & 251.4k   & 2   & 168  & 5.6$\cdot 10^{-2}$ & 214  & 6.9$\cdot 10^{-3}$ & 0.914 \\
\textsc{slashdot}       & 82.1k  & 0.5M     & 252 & 334  & 4.8$\cdot 10^{-2}$ & 400  & 6.6$\cdot 10^{-4}$ & 0.953 \\
\textsc{epinions}       & 131.6k & 0.71M    & 481 & 243  & 9.0$\cdot 10^{-2}$ & 596  & 2.0$\cdot 10^{-2}$ & 0.141 \\
\textsc{politics}       & 138k   & 0.72M    & 130 & 431  & 1.8$\cdot 10^{-1}$ & 556  & 1.5$\cdot 10^{-2}$ & 0.798 \\
\textsc{conflict}       & 116.7k & 2.03M    & 964 & 1344 & 4.3$\cdot 10^{-2}$ & 1628 & 1.8$\cdot 10^{-4}$ & 0.961 \\
\hline
\end{tabular}
\end{table*}

\begin{table}[h]
\vspace{-0.2cm}
	\centering
    \caption{MBS retrieved from real-world datasets: $|V|$, $|E|$, ratio of negative edges $\rho^{-}$ and ratio of non-zero elements $\delta$}
    \label{tab:MBS}
	\begin{tabular}{ccccc}
		\hline
		Dataset & $|V|$ & $|E|$ & $\rho^-$ & $\delta$ \\
        \hline 
\textsc{Bitcoin} &  5k & 12k & 0.09 & 1.1e-03 \\
\textsc{Epinions} &  85k & 307k & 0.07 & 8.5e-05 \\
\textsc{Slashdot} &  56k & 177k & 0.04 & 1.1e-04 \\
\textsc{Twitter} &  10k & 210k & 0.001 & 4.5e-03 \\
\textsc{Conflict} &  67k & 808k & 0.64 & 3.6e-04 \\
\textsc{Elections} &  4k & 24k & 0.07 & 3.3e-03 \\
\textsc{Politics} &  74k & 322k & 0.04 & 1.2e-04 \\
\hline
	\end{tabular}

\end{table}

\subsection{Effectiveness}

We demonstrate that balance rate behaves consistently with structural balance theory: it equals one for perfectly balanced graphs, approaches zero for highly unbalanced structures, and varies smoothly as imbalance increases, offering an intuitive quantitative notion of balance under uncertainty.
Starting from a balanced graph, we progressively introduce imbalance by adding edges that violate structural balance. As the fraction of added critical edges increases, the balance rate decreases monotonically. Separately, uniformly increasing the existence probabilities of all edges makes the realized graph denser, causing latent inconsistencies—whether introduced or original—to appear more frequently and jointly. For any fixed level of structural corruption, higher edge probabilities therefore lead to lower balance rates. Taken together, balance rate varies smoothly across both dimensions, attaining high values only when imbalance is structurally scarce and probabilistically suppressed, and dropping sharply when imbalance becomes both widespread and likely. This consistent behavior confirms that the balance rate provides a sensitive and intuitive measure of balance under uncertainty.
We additionally rigorously formalize this intuition in Lemma~\ref{p_eta_monotonicity} and Corollary~\ref{cor_monotonicity}. Lemma~\ref{p_eta_monotonicity} is proved in Appendix~\ref{app:missing_proofs}.

\begin{lemma}\label{p_eta_monotonicity}
	Let $G = (V, E, S)$ be a signed graph, $P, Q$ be two probability distributions on $E$ satisfying $p_e \leq q_e$ for any $e \in E$. Then 
	$$\mathbb{E}_{P}[\mathbb{I}(G \text{ is balanced})] \geq \mathbb{E}_{Q}[\mathbb{I}(G \text{ is balanced})].$$
\end{lemma}

\begin{corollary}\label{cor_monotonicity}
	Let $G_1 = (V, E_1, S_1, P_1), G_2 = (V, E_2, S_2, P_2)$ be USGs with $E_1 \subseteq E_2$, $S_2|_{E_1} = S_1$, $P_2|_{E_1} = P_1$. Then $R_{bal}(G_1) \geq R_{bal}(G_2)$.
\end{corollary}

For each dataset, we use the algorithm by~\cite{chen2024scalable} to extract a large Maximum Balanced Subgraph (MBS), summarized in \Cref{tab:MBS}. Then we vary two parameters: the fraction of injected critical edges $\eta$, ranging from $0.05\%$ to $5\%$, and a global probability scaling factor $p_{mul}$, ranging from 1.0 to 10, which uniformly multiplies the existence probabilities of all edges. For each $(\eta, p_{mul})$ pair, we compute the corresponding balance rate. The results are summarized in Figure~\ref{plot:p_eta}.

\begin{figure*}[t]
\centering
\includegraphics[width=0.36\textwidth]{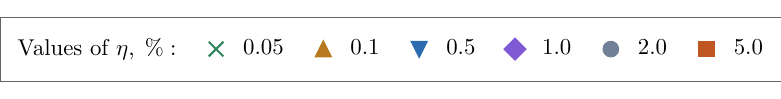}\\[0.3em]

\begin{tabular}{ccc}
\includegraphics[width=0.297\textwidth]{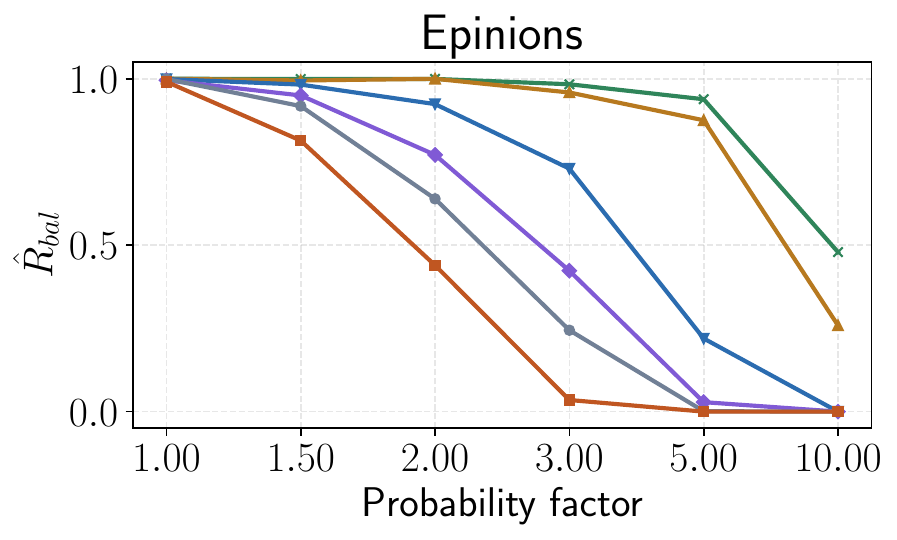} &
\includegraphics[width=0.297\textwidth]{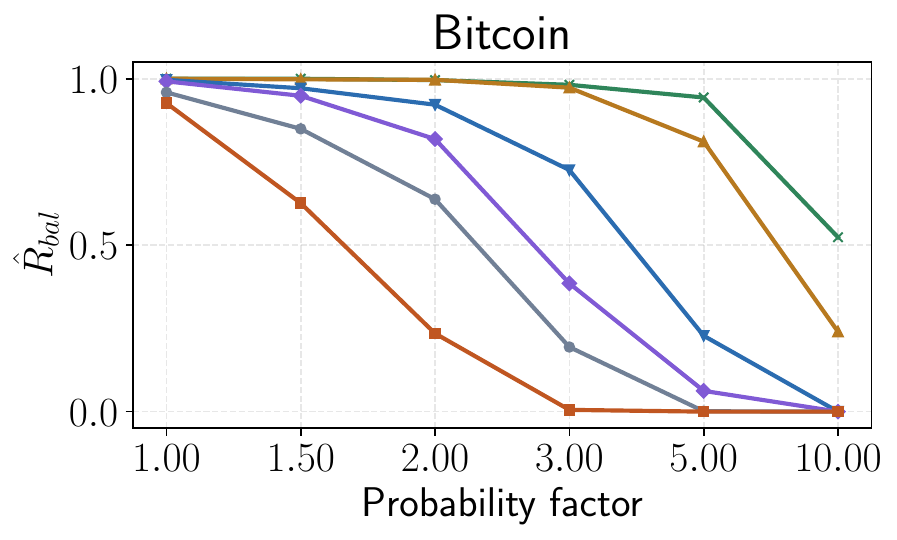} &
\includegraphics[width=0.297\textwidth]{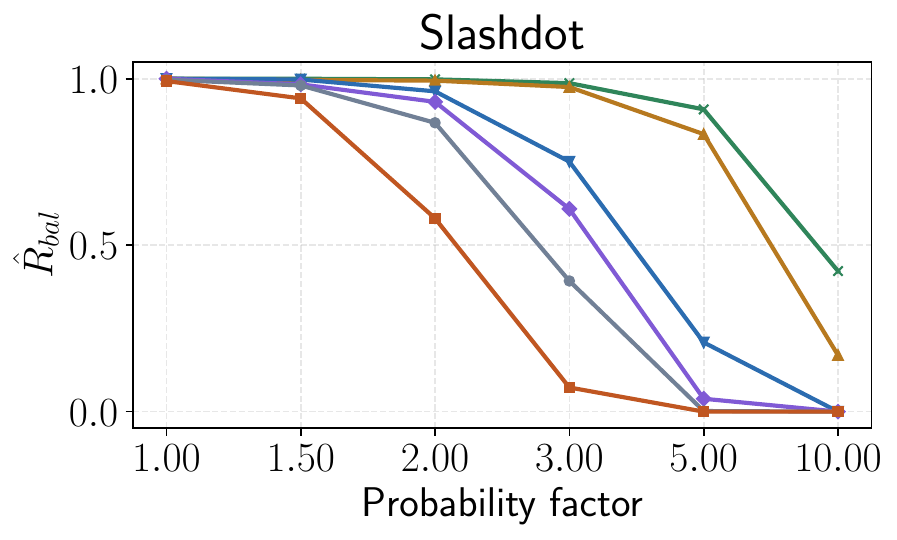} \\

\includegraphics[width=0.297\textwidth]{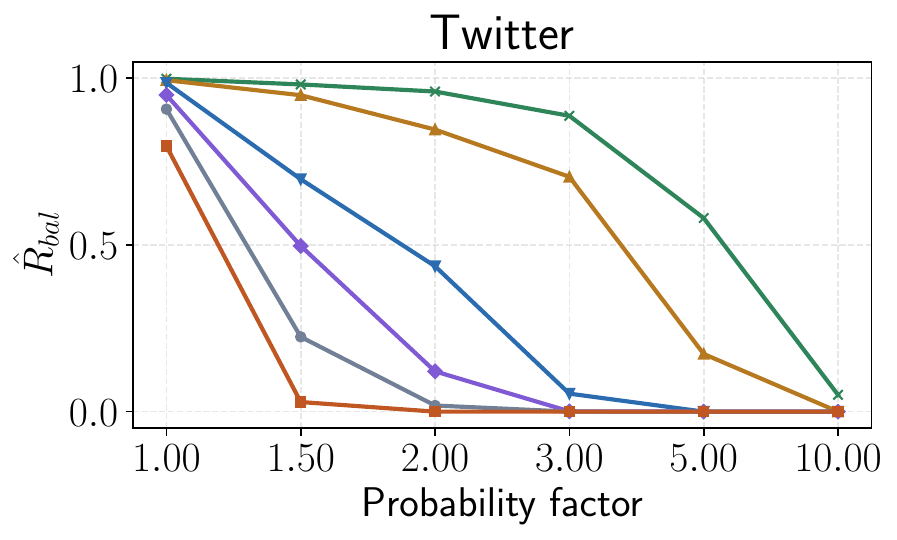} &
\includegraphics[width=0.297\textwidth]{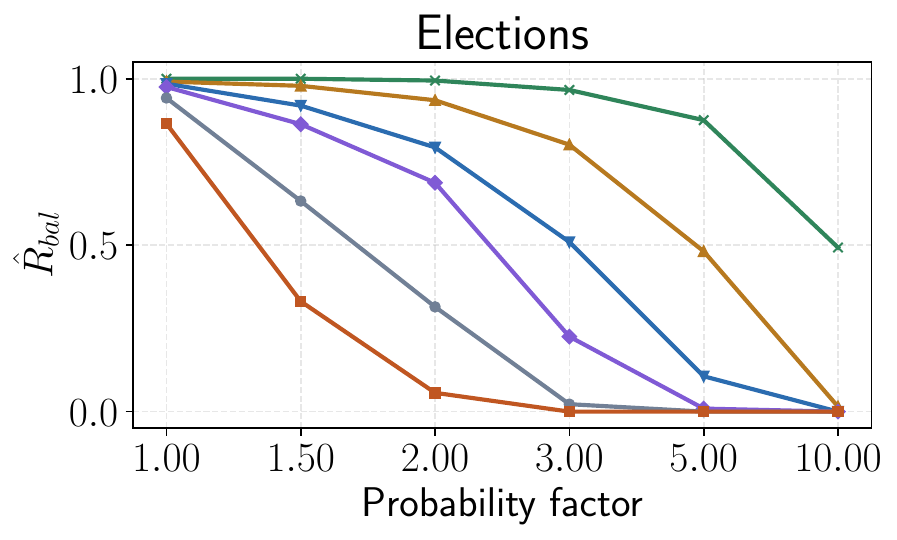} &
\includegraphics[width=0.297\textwidth]{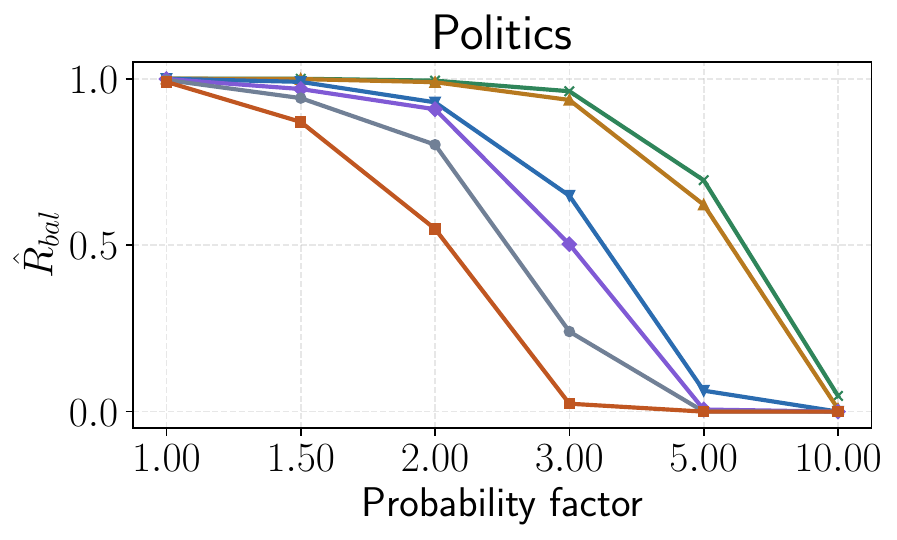}
\end{tabular}

\vspace{-0.3cm}
\caption{Comparison of balance rate for different $p_{mul}, \eta$.}
\label{plot:p_eta}
\end{figure*}

\subsection{Applicability}

We illustrate that the balance rate can be leveraged as a sensitivity measure for identifying structurally critical edges, enabling effective detection of edges whose presence disproportionately affects global balance. To demonstrate the metric's utility in isolating these edges, we designed a controlled "needle in a haystack" experiment involving ground-truth recovery.

\vspace{0.25cm}
\noindent \textit{Experimental Setup: The Hidden Conflict Scenario.}
For each real-world dataset described above, we simulated a scenario where a latent balanced structure is obscured by a small set of conflicting connections. The procedure was as follows:

\begin{enumerate}[leftmargin=*]
    \item \textbf{Ground Truth Extraction:} Same as above, we use the extracted balanced graphs, denoted as $G_{bal} = (V, E_{bal})$ (\Cref{tab:MBS}) as a clean, structurally balanced baseline.

    \item \textbf{Conflict Injection:} We injected a set of 5 \textit{critical edges}, denoted as $E_{crit}$, into $G_{bal}$. These edges were specifically selected to create polarity violations (unbalanced cycles), thereby reintroducing structural frustration into the network.
    
    \item \textbf{Candidate Set Generation:} To mask the critical edges, we combined them with 95 benign edges randomly sampled from the original graph structure. This formed a \textit{candidate set} $S$ of 100 edges, where $E_{crit} \subset S$ and $|S|=100$.
\end{enumerate}

The challenge is to recover the subset $E_{crit}$ from $S$ purely by optimizing the Balance Rate without prior knowledge.

\vspace{0.25cm}
\noindent \textit{Methodology: Greedy Optimization.}
We modeled the edge existence in the candidate set as a random signed graph process. To incorporate uncertainty, the existence probability $p_e$ for every edge $e \in S$ was drawn uniformly from the interval $[0.5, 1]$. We used 200 samples for each evaluation of the balance rate. We employed a greedy erasure strategy: at each step, we removed the edge $e \in S$ whose removal maximized the balance rate. We repeated this process for 5 iterations to identify the top-5 most critical edges.

\vspace{0.25cm}
\noindent \textit{Results: Precision in Conflict Detection.}
The experiment was conducted over $100$ independent trials (varying random edge probabilities and injected edge locations) for each dataset.
The results demonstrate that the Balance Rate is a highly discriminative filter for structural conflict. In \textbf{100\%} of the cases across all datasets, the 5 edges identified by the greedy erasure corresponded exactly to the injected set $E_{crit}$. This indicates that the Balance Rate can effectively pinpoint local structural anomalies, separating signal (critical conflicts) from noise (benign edges).

%% file: conclusion.tex
\section{CONCLUSION}\label{sec:conclusion}

In this work, we introduced the balance rate as a principal metric to quantify structural stability in uncertain signed graphs and established its theoretical and computational foundations. We proved that exact computation of the balance rate is NP-hard, clarifying the inherent difficulty of balance evaluation under uncertainty and motivating the need for approximation methods. To address this challenge, we developed a scalable estimation framework that combines structural graph decomposition with a Rao–Blackwellized spanning-tree estimator. This design reduces variance through principled conditioning and exploits the independence structure induced by biconnected components, yielding near-linear per-sample complexity while maintaining unbiasedness. We further provided asymptotically justified confidence intervals via the Delta method, enabling statistically sound inference for large-scale networks. Extensive experiments confirmed that our approach is both computationally efficient and structurally meaningful: the balance rate responds consistently to controlled perturbations and reveals edges that play a critical role in destabilizing network structure. Our framework opens several promising research directions, including developing efficient algorithms for detecting critical substructures.

%% file: appendix.tex
\appendix

\section{OMITTED PROOFS}\label{app:missing_proofs}

\subsection{Proof of Lemma~\ref{lem:bridge_decomp}}

Since $e$ is a bridge, it does not belong to any simple cycle of $G$. Hence,
for any realization $\omega$, the presence or absence of $e$ does not
affect whether $G(\omega)$ is balanced.

Moreover, $G(\omega)$ is balanced if and only if both subgraphs
$G_1(\omega)$ and $G_2(\omega)$ are balanced. Because the edge sets of
$G_1$ and $G_2$ are disjoint and edge realizations are independent,
these events are independent. Therefore,

\begin{align*}
R_{bal}(G)
&=
\mathbb{P}(G_1(\omega)\text{ balanced})\,
\mathbb{P}(G_2(\omega)\text{ balanced})
= \\ &=
R_{bal}(G_1)\,R_{bal}(G_2).
\end{align*}

\subsection{Proof of Lemma~\ref{lem:cutpoint_decomp}}

Every cycle of $G$ is contained entirely within exactly one of the
subgraphs $G_i$. Consequently, a realization $G(\omega)$ is balanced
if and only if each $G_i(\omega)$ is balanced.

Although the subgraphs share the articulation vertex $v$, their edge
sets are disjoint. Since edge realizations are independent, the balance
events of $G_1(\omega),\ldots,G_k(\omega)$ are independent, yielding
\[
R_{bal}(G)
=
\prod_{i=1}^k
\mathbb{P}\bigl(G_i(\omega)\text{ is balanced}\bigr)
=
\prod_{i=1}^k R_{bal}(G_i).
\]

\subsection{Proof of \Cref{thm:variance-reduction}}
    Since $\widehat R_{\mathrm{joint}}$ is an average of i.i.d.\ samples,
    \[
    \mathrm{Var}(\widehat R_{\mathrm{joint}})
    =
    \frac{1}{N}
    \left(
    \prod_{j=1}^m \mathbb{E}[R_{bal,j}^2]
    -
    \prod_{j=1}^m \mu_j^2
    \right).
    \]
    Define
    \[
    a_j := \mu_j^2,
    \qquad
    b_j := \mathbb{E}[R_{bal,j}^2]-\mu_j^2
    = \mathrm{Var}(R_{bal,j}) \ge 0.
    \]
    Then
    \begin{equation}
        \mathrm{Var}(\widehat R_{\mathrm{joint}})
        =
        \frac{1}{N}
        \left(
        \prod_{j=1}^m (a_j+b_j)
        -
        \prod_{j=1}^m a_j
        \right).
        \label{eq:var-joint}
    \end{equation}
    
    Next, since the $\bar R_{bal,j}$ are independent,
    \[
    \mathrm{Var}(\widehat R_{\mathrm{prod}})
    =
    \prod_{j=1}^m \mathbb{E}[\bar R_{bal,j}^2]
    -
    \prod_{j=1}^m \mu_j^2.
    \]
    A direct computation yields
    \[
    \mathbb{E}[\bar R_{\mathrm{bal},j}^2]
    =
    \mu_j^2 + \frac{b_j}{N},
    \]
    and therefore
    \begin{equation}
        \mathrm{Var}(\widehat R_{\mathrm{prod}})
        =
        \prod_{j=1}^m\left(a_j+\frac{b_j}{N}\right)
        -
        \prod_{j=1}^m a_j.
        \label{eq:var-prod}
    \end{equation}
    
    Expanding the products in
    \eqref{eq:var-joint} and \eqref{eq:var-prod} over subsets
    $S\subseteq\{1,\dots,m\}$ gives
    \[
    \mathrm{Var}(\widehat R_{\mathrm{joint}})
    =
    \sum_{S\neq\emptyset}
    \frac{1}{N}
    \left(\prod_{j\in S} b_j\right)
    \left(\prod_{j\notin S} a_j\right),
    \]
    while
    \[
    \mathrm{Var}(\widehat R_{\mathrm{prod}})
    =
    \sum_{S\neq\emptyset}
    \frac{1}{N^{|S|}}
    \left(\prod_{j\in S} b_j\right)
    \left(\prod_{j\notin S} a_j\right).
    \]
    
    Since $N^{-|S|}\le N^{-1}$ for all $|S|\ge 1$, the inequality follows termwise.
    Strict inequality holds unless $b_j=0$ for all $j$, i.e., unless each
    $\hat{R}_{bal,j}$ is deterministic.

\subsection{Proof of Lemma~\ref{p_eta_monotonicity}}

We employ a standard coupling argument to establish the stochastic dominance. Let $(\Omega, \mathcal{F}, \mathbb{P})$ be a probability space supporting a collection of independent random variables $\{U_e\}_{e \in E}$, where each $U_e$ is uniformly distributed on the interval $[0, 1]$.

We define the coupled random graphs $G_P$ and $G_Q$ on this common probability space via the indicator functions for their edge sets. For any edge $e \in E$, let:
\begin{align*}
	\mathbf{I}(e \in E(G_P)) &= \mathbf{I}_{\{U_e \le p_e\}}, \\
	\mathbf{I}(e \in E(G_Q)) &= \mathbf{I}_{\{U_e \le q_e\}}.
\end{align*}
By the hypothesis $p_e \le q_e$, the event $\{U_e \le p_e\}$ implies $\{U_e \le q_e\}$. Consequently, for every realization $\omega \in \Omega$, the edge set of $G_P$ is a subset of the edge set of $G_Q$:
\begin{equation*}
	E(G_P(\omega)) \subseteq E(G_Q(\omega))
\end{equation*}
Thus, $G_P(\omega)$ is a subgraph of $G_Q(\omega)$ almost surely.

A signed graph is structurally balanced if and only if it contains no cycles with an odd number of negative edges (negative cycles). Let $\mathcal{C}(H)$ denote the set of all cycles in a graph $H$. The condition for a graph $H$ to be balanced is:
\begin{equation*}
	\begin{aligned}
		\forall C \in \mathcal{C}(H), \quad \prod_{e \in C} \sigma(e) = +1.
	\end{aligned}
\end{equation*}
Since $G_P(\omega) \subseteq G_Q(\omega)$, any cycle present in $G_P(\omega)$ must also exist in $G_Q(\omega)$, i.e., $\mathcal{C}(G_P(\omega)) \subseteq \mathcal{C}(G_Q(\omega))$.

If $G_Q(\omega)$ is balanced, then all cycles in $\mathcal{C}(G_Q(\omega))$ are positive. Since $\mathcal{C}(G_P(\omega))$ is a subset of these cycles, all cycles in $G_P(\omega)$ must also be positive. Therefore, the property of being balanced $ \mathcal{B}$ is monotone decreasing with respect to edge inclusion:
\begin{equation*}
	G_Q(\omega) \in \mathcal{B} \implies G_P(\omega) \in \mathcal{B}.
\end{equation*}
Taking expectations over the probability space $\Omega$, we obtain:
\begin{equation*}
	\begin{aligned}
		\mathbb{P}(G_Q \in \mathcal{B}) &= \mathbb{P}(\{\omega : G_Q(\omega) \in \mathcal{B}\}) \leq \mathbb{P}(\{\omega : G_P(\omega) \in \mathcal{B}\}) = \\ &=\mathbb{P}(G_P \in \mathcal{B}).
	\end{aligned}
\end{equation*}

\begin{figure*}[t]
\centering

\includegraphics[width=0.33\textwidth]{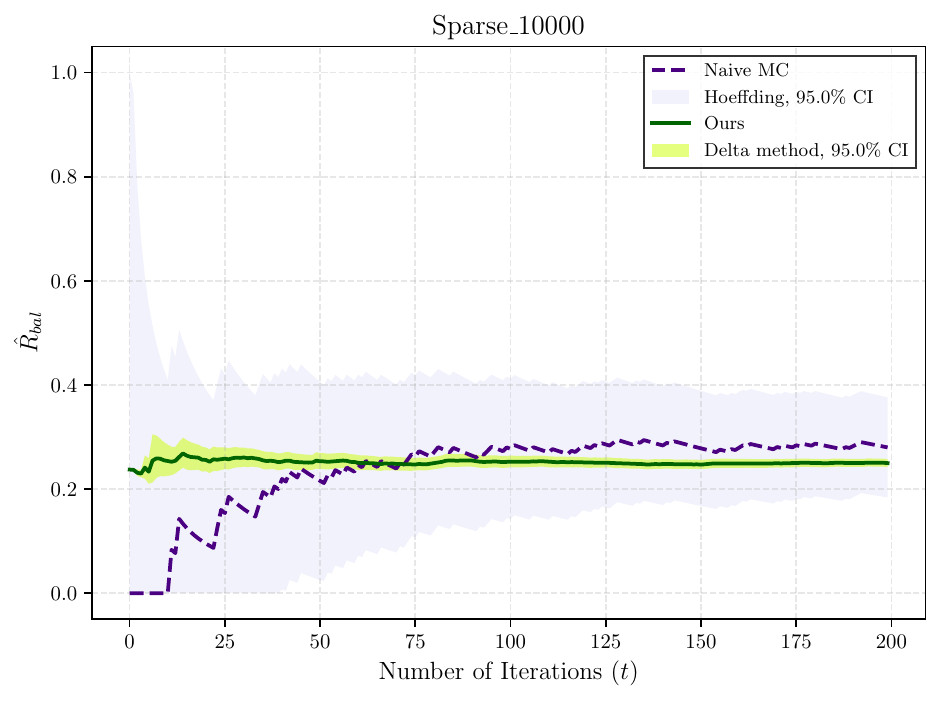}\hfill
\includegraphics[width=0.33\textwidth]{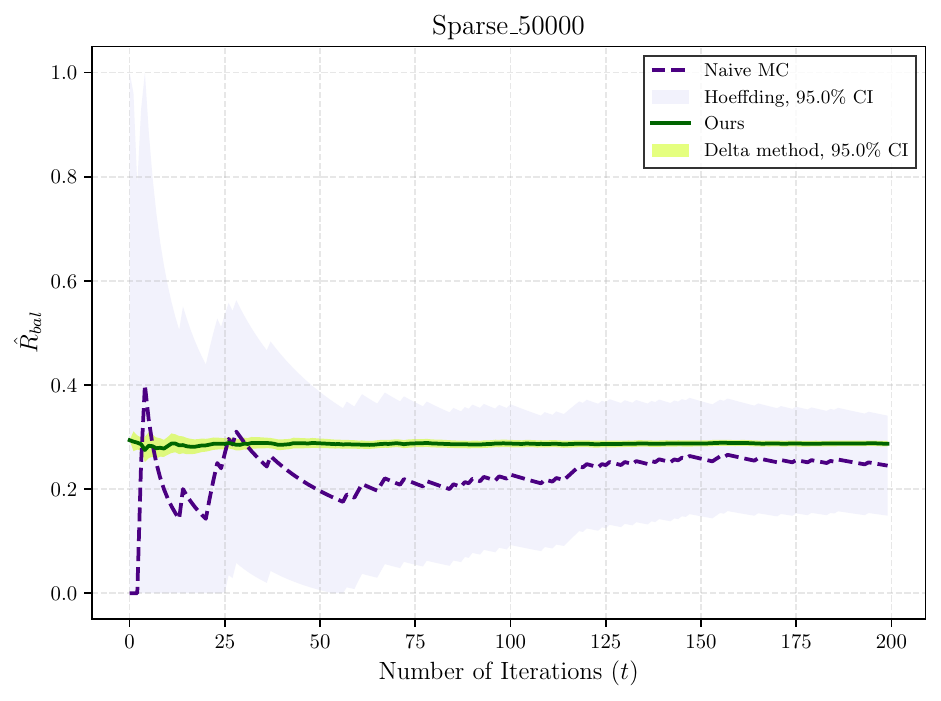}\hfill
\includegraphics[width=0.33\textwidth]{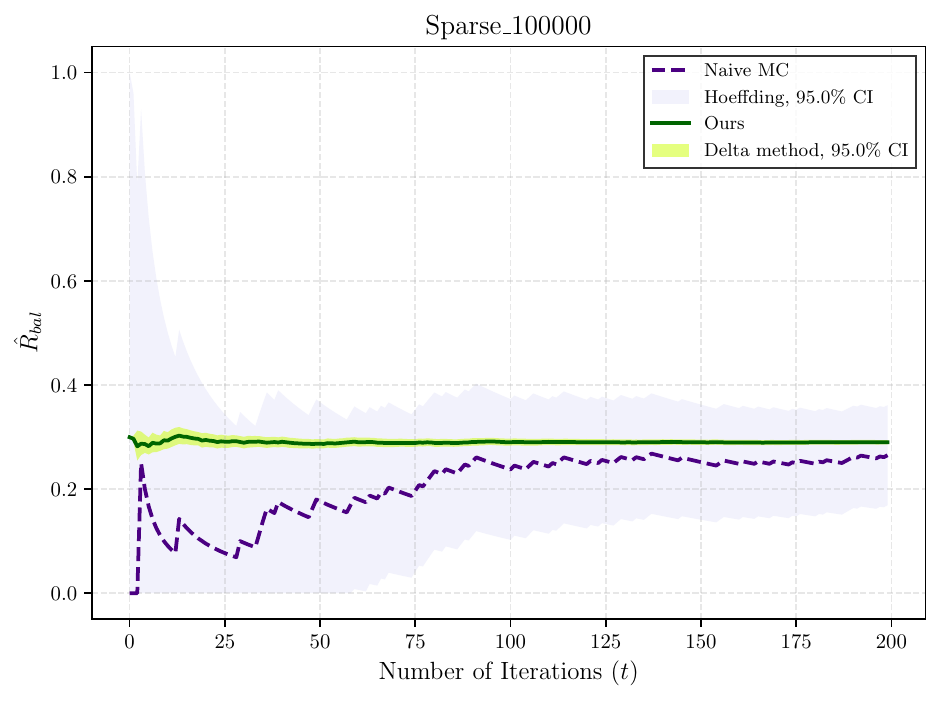}\hfill
\includegraphics[width=0.33\textwidth]{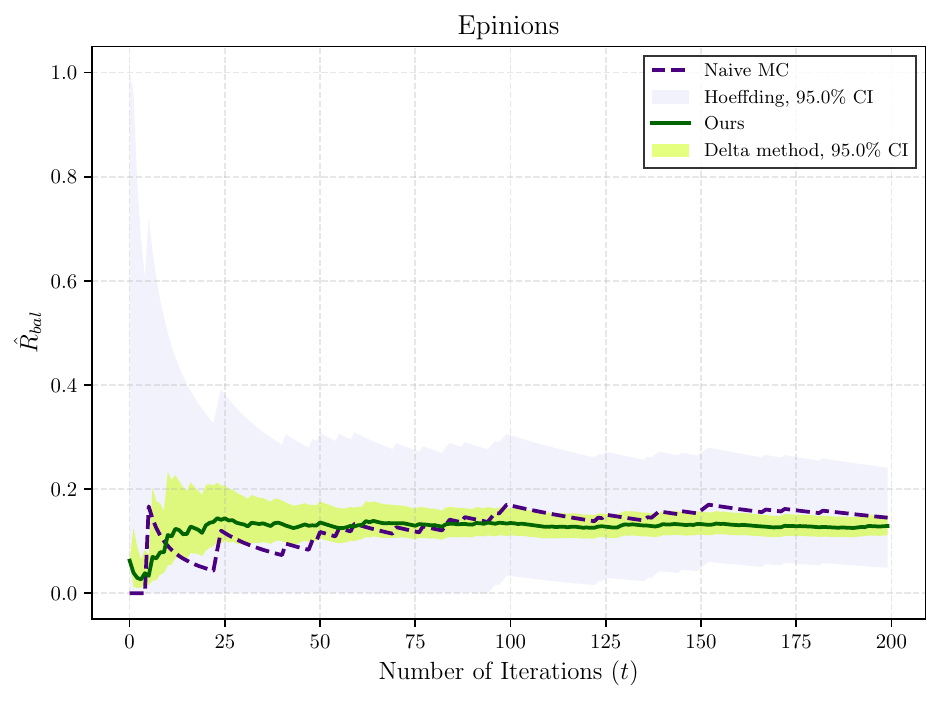}\hfill
\includegraphics[width=0.33\textwidth]{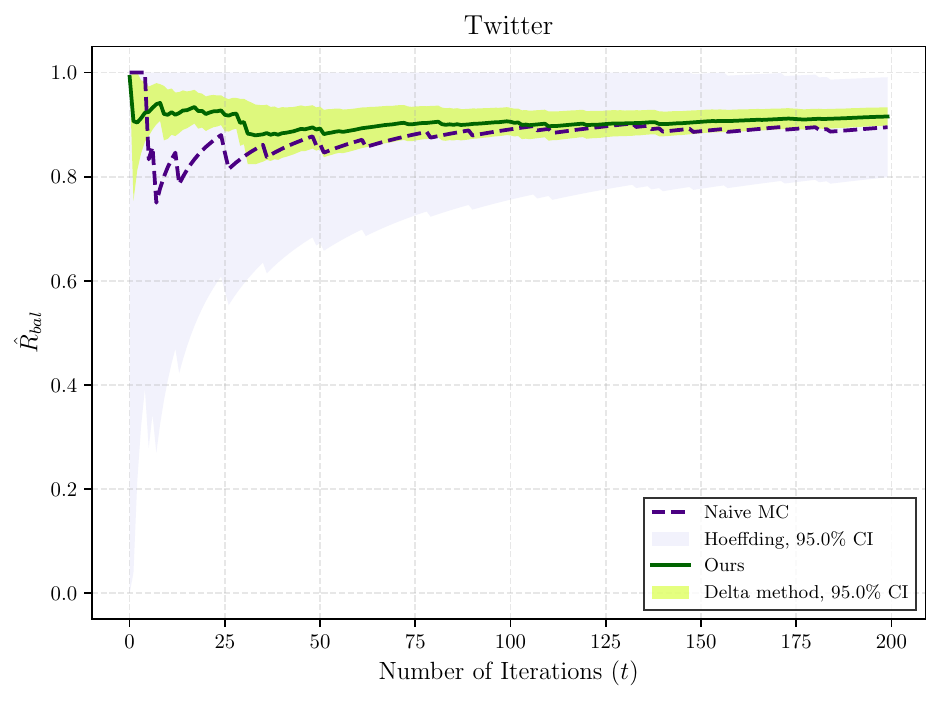}\hfill
\includegraphics[width=0.33\textwidth]{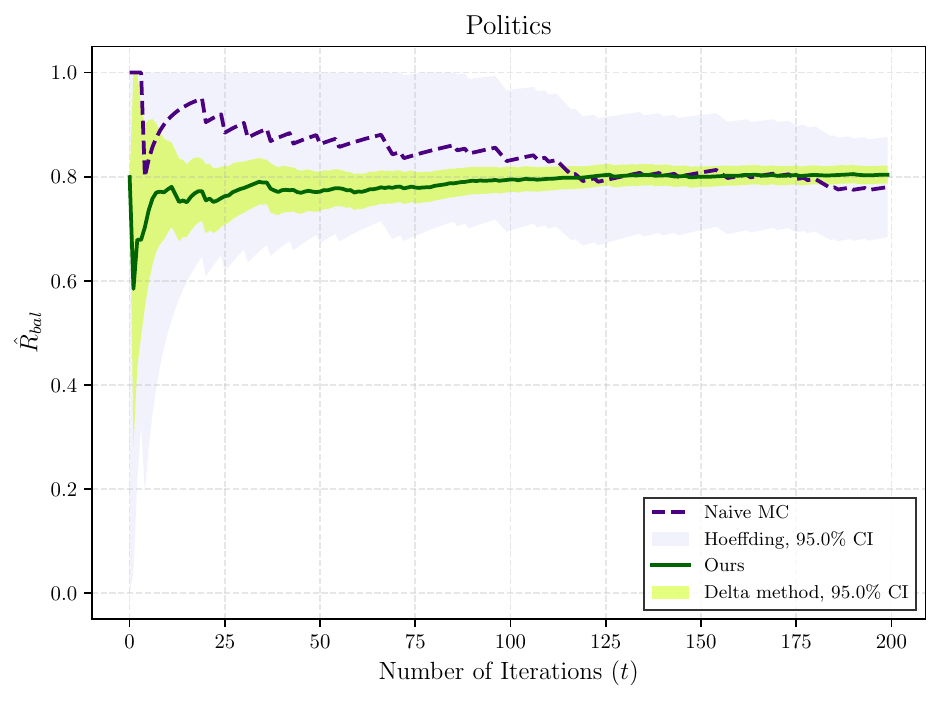}\hfill
\caption{Comparison of prefix confidence intervals between naive Monte Carlo (MC) sampling and Rao-Blackwellized (RB) sampling}\label{plot:conf_intervals}
\end{figure*}

\section{Balance Rate Computation}

Similar to certain graph balance detection, in USGs, balance rate is also related to negative cycles.

\begin{theorem}\label{thm:balance_rate_formula}
Let $\mathcal{NC}$ denotes the set of all negative cycles of a USG $G$, and $Q_i$ denotes the probability that $NC_i$ does not exist ($NC_i\in\mathcal{NC}$) respectively, i.e., 
\[
Q_i=1-\prod_{e\in NC_i}p_e.
\]
Then the balance rate of $G$ is the probability that none of the negative cycles in $\mathcal{NC}$ exist, i.e.,
\[
    R_{bal}(G)=\Pr(\bigcap_{NC_i\in\mathcal{NC}} Q_i).
\]
\end{theorem}

Based on \Cref{thm:balance_rate_formula}, the global balance rate can be calculated separately, which inspires Lemmas \ref{lem:bridge_decomp} and \ref{lem:cutpoint_decomp}, both of which decompose a connected graph into two connected components. 

\section{CONFIDENCE INTERVALS COMPARISON}
Shown in \Cref{plot:conf_intervals}, on both synthetic and real-world datasets, our proposed algorithm has tighter confidence intervals than the naive MC sampling for one magnitude. This aligns with our theoretical analysis in \Cref{sec:delta_method_bounds}.

%% file: ref.bib
@article{cartwright1956,
	author    = {Cartwright, D. and Harary, F.},
	title     = {Structural balance: a generalization of Heider's theory},
	journal   = {Psychological Review},
	volume    = {63},
	pages     = {277--293},
	year      = {1956}
}

@article{zhang2017,
	author    = {Zhang, Y. and Tang, J. and Liu, H.},
	title     = {Mining Signed Networks in Social Media},
	journal   = {ACM SIGKDD Explorations},
	volume    = {19},
	number    = {2},
	pages     = {1--20},
	year      = {2017}
}

@article{doreian2019,
	author    = {Doreian, P.},
	title     = {Structural Balance and Signed Networks in International Relations},
	journal   = {Social Networks},
	volume    = {57},
	pages     = {89--100},
	year      = {2019}
}

@article{costanzo2016,
	author    = {Costanzo, M. et al.},
	title     = {A global genetic interaction network maps a wiring diagram of cellular function},
	journal   = {Science},
	volume    = {353},
	number    = {6306},
	pages     = {1423--1431},
	year      = {2016}
}

@article{liu2015,
	author    = {Liu, X. et al.},
	title     = {Analysis of Balanced Genetic Networks},
	journal   = {Bioinformatics},
	volume    = {31},
	number    = {3},
	pages     = {394--402},
	year      = {2015}
}

@article{saiz2018,
	author    = {Saiz, H. et al.},
	title     = {Structural Balance in Plant Communities},
	journal   = {Ecology Letters},
	volume    = {21},
	pages     = {199--208},
	year      = {2018}
}

@article{edge_deletion,
	author = {Yannakakis, Mihalis},
	title = {Edge-Deletion Problems},
	journal = {SIAM Journal on Computing},
	volume = {10},
	number = {2},
	pages = {297-309},
	year = {1981},
	doi = {10.1137/0210021},
	URL = { https://doi.org/10.1137/0210021}
}

@book{casella2002statistical,
	title={Statistical Inference},
	author={Casella, George and Berger, Roger L.},
	year={2002},
	publisher={Duxbury Press}
}

@article{galil1991data,
	title={Data Structures and Algorithms for Disjoint Set Union Problems},
	author={Galil, Zvi and Italiano, Giuseppe F.},
	journal={ACM Computing Surveys},
	volume={23},
	number={3},
	pages={319–344},
	year={1991},
	publisher={ACM}
}

@book{wasserman2004all,
	title={All of statistics: a concise course in statistical inference},
	author={Wasserman, Larry},
	year={2004},
	publisher={Springer Science \& Business Media}
}

@article{zhang2024finding,
  title={Finding Antagonistic Communities in Signed Uncertain Graphs},
  author={Zhang, Qiqi and Chu, Lingyang and Zhao, Zijin and Pei, Jian},
  journal={IEEE Transactions on Knowledge and Data Engineering},
  year={2024},
  publisher={IEEE}
}

@article{harary1967graph,
  title={Graph theory and theoretical physics},
  author={Harary, Frank and others},
  journal={(No Title)},
  year={1967}
}

@article{aggarwal1975reliability,
  title={Reliability evaluation a comparative study of different techniques},
  author={Aggarwal, KK and Misra, KB and Gupta, JS},
  journal={Microelectronics Reliability},
  volume={14},
  number={1},
  pages={49--56},
  year={1975},
  publisher={Elsevier}
}

@article{khan_conditional_2018,
	title = {Conditional {Reliability} in {Uncertain} {Graphs}},
	issn = {1041-4347},
	url = {http://ieeexplore.ieee.org/document/8318619/},
	doi = {10.1109/TKDE.2018.2816653},
	language = {en},
	urldate = {2022-11-30},
	journal = {IEEE Transactions on Knowledge and Data Engineering},
	author = {Khan, Arijit and Bonchi, Francesco and Gullo, Francesco and Nufer, Andreas},
	year = {2018},
	pages = {1--1},
}

@inproceedings{potamias_fast_2009,
	address = {Hong Kong, China},
	title = {Fast shortest path distance estimation in large networks},
	isbn = {978-1-60558-512-3},
	url = {http://portal.acm.org/citation.cfm?doid=1645953.1646063},
	doi = {10.1145/1645953.1646063},
	booktitle = {Proceeding of the 18th {ACM} conference on {Information} and knowledge management - {CIKM} '09},
	publisher = {ACM Press},
	author = {Potamias, Michalis and Bonchi, Francesco and Castillo, Carlos and Gionis, Aristides},
	year = {2009},
	pages = {867},
}

@article{jin_distance-constraint_2011,
	title = {Distance-constraint reachability computation in uncertain graphs},
	volume = {4},
	issn = {2150-8097},
	url = {https://dl.acm.org/doi/10.14778/2002938.2002941},
	doi = {10.14778/2002938.2002941},
	journal = {Proceedings of the VLDB Endowment},
	author = {Jin, Ruoming and Liu, Lin and Ding, Bolin and Wang, Haixun},
	month = jun,
	year = {2011},
	pages = {551--562},
}

@inproceedings{li_discovering_2017,
	address = {Chicago Illinois USA},
	title = {Discovering {Your} {Selling} {Points}: {Personalized} {Social} {Influential} {Tags} {Exploration}},
	isbn = {978-1-4503-4197-4},
	shorttitle = {Discovering {Your} {Selling} {Points}},
	doi = {10.1145/3035918.3035952},
	booktitle = {Proceedings of the 2017 {ACM} {International} {Conference} on {Management} of {Data}},
	publisher = {ACM},
	author = {Li, Yuchen and Fan, Ju and Zhang, Dongxiang and Tan, Kian-Lee},
	month = may,
	year = {2017},
	pages = {619--634},
}

@inproceedings{jin2010computing,
  title={Computing label-constraint reachability in graph databases},
  author={Jin, Ruoming and Hong, Hui and Wang, Haixun and Ruan, Ning and Xiang, Yang},
  booktitle={Proceedings of the 2010 ACM SIGMOD International Conference on Management of data},
  pages={123--134},
  year={2010}
}

@article{fishman_comparison_1986,
	title = {A {Comparison} of {Four} {Monte} {Carlo} {Methods} for {Estimating} the {Probability} of s-t {Connectedness}},
	volume = {35},
	issn = {0018-9529},
	doi = {10.1109/TR.1986.4335388},
	journal = {IEEE Transactions on Reliability},
	author = {Fishman, George S.},
	year = {1986},
	pages = {145--155},
}

@article{li_recursive_2016,
	title = {Recursive {Stratified} {Sampling}: {A} {New} {Framework} for {Query} {Evaluation} on {Uncertain} {Graphs}},
	volume = {28},
	issn = {1041-4347},
	shorttitle = {Recursive {Stratified} {Sampling}},
	url = {http://ieeexplore.ieee.org/document/7286806/},
	doi = {10.1109/TKDE.2015.2485212},
	journal = {IEEE Transactions on Knowledge and Data Engineering},
	author = {Li, Rong-Hua and Yu, Jeffrey Xu and Mao, Rui and Jin, Tan},
	month = feb,
	year = {2016},
	pages = {468--482},
}

@article{maniu_indexing_2017,
	title = {An {Indexing} {Framework} for {Queries} on {Probabilistic} {Graphs}},
	volume = {42},
	issn = {0362-5915, 1557-4644},
	doi = {10.1145/3044713},
	journal = {ACM Transactions on Database Systems},
	author = {Maniu, Silviu and Cheng, Reynold and Senellart, Pierre},
	month = jun,
	year = {2017},
	pages = {1--34},
}

@inproceedings{wang2024fast,
  title={Fast Query Answering by Labeling Index on Uncertain Graphs},
  author={Wang, Zeyu and Shi, Qihao and Chen, Jiawei and Wang, Can and Song, Mingli and Wang, Xinyu},
  booktitle={2024 IEEE 40th International Conference on Data Engineering (ICDE)},
  pages={4058--4071},
  year={2024},
  organization={IEEE}
}

@article{ball_computational_1986,
	title = {Computational {Complexity} of {Network} {Reliability} {Analysis}: {An} {Overview}},
	volume = {35},
	issn = {0018-9529},
	shorttitle = {Computational {Complexity} of {Network} {Reliability} {Analysis}},
	url = {http://ieeexplore.ieee.org/document/4335422/},
	doi = {10.1109/TR.1986.4335422},
	journal = {IEEE Transactions on Reliability},
	author = {Ball, Michael O.},
	year = {1986},
	pages = {230--239},
}

@article{valiant1979complexity,
  title={The complexity of enumeration and reliability problems},
  author={Valiant, Leslie G},
  journal={siam Journal on Computing},
  volume={8},
  number={3},
  pages={410--421},
  year={1979},
  publisher={SIAM}
}

@inproceedings{zou2013polynomial,
  title={Polynomial-time algorithm for finding densest subgraphs in uncertain graphs},
  author={Zou, Zhaonian},
  booktitle={Proceedings of MLG Workshop},
  year={2013}
}

@inproceedings{bonchi2014core,
  title={Core decomposition of uncertain graphs},
  author={Bonchi, Francesco and Gullo, Francesco and Kaltenbrunner, Andreas and Volkovich, Yana},
  booktitle={Proceedings of the 20th ACM SIGKDD international conference on Knowledge discovery and data mining},
  pages={1316--1325},
  year={2014}
}

@inproceedings{huang2016truss,
  title={Truss decomposition of probabilistic graphs: Semantics and algorithms},
  author={Huang, Xin and Lu, Wei and Lakshmanan, Laks VS},
  booktitle={Proceedings of the 2016 International Conference on Management of Data},
  pages={77--90},
  year={2016}
}

@article{zou2017truss,
  title={Truss decomposition of uncertain graphs},
  author={Zou, Zhaonian and Zhu, Rong},
  journal={Knowledge and Information Systems},
  volume={50},
  pages={197--230},
  year={2017},
  publisher={Springer}
}

@article{yuan2012efficient,
  title={Efficient Subgraph Similarity Search on Large Probabilistic Graph Databases},
  author={Yuan, Ye and Wang, Guoren and Chen, Lei and Wang, Haixun},
  journal={Proceedings of the VLDB Endowment},
  year={2012}
}

@inproceedings{papapetrou2011efficient,
  title={Efficient discovery of frequent subgraph patterns in uncertain graph databases},
  author={Papapetrou, Odysseas and Ioannou, Ekaterini and Skoutas, Dimitrios},
  booktitle={Proceedings of the 14th International Conference on Extending Database Technology},
  pages={355--366},
  year={2011}
}

@article{wen2020computing,
  title={Computing k-cores in large uncertain graphs: An index-based optimal approach},
  author={Wen, Dong and Yang, Bohua and Qin, Lu and Zhang, Ying and Chang, Lijun and Li, Rong-Hua},
  journal={IEEE Transactions on Knowledge and Data Engineering},
  volume={34},
  number={7},
  pages={3126--3138},
  year={2020},
  publisher={IEEE}
}

@inproceedings{yang2019index,
  title={Index-based optimal algorithm for computing k-cores in large uncertain graphs},
  author={Yang, Bohua and Wen, Dong and Qin, Lu and Zhang, Ying and Chang, Lijun and Li, Rong-Hua},
  booktitle={2019 IEEE 35th International Conference on Data Engineering (ICDE)},
  pages={64--75},
  year={2019},
  organization={IEEE}
}

@article{harary1953notion,
  title={On the notion of balance of a signed graph.},
  author={Harary, Frank},
  journal={Michigan Mathematical Journal},
  volume={2},
  number={2},
  pages={143--146},
  year={1953},
  publisher={University of Michigan, Department of Mathematics}
}

@article{harary1980simple,
  title={A simple algorithm to detect balance in signed graphs},
  author={Harary, Frank and Kabell, Jerald A},
  journal={Mathematical Social Sciences},
  volume={1},
  number={1},
  pages={131--136},
  year={1980},
  publisher={Elsevier}
}

@article{aref2020modeling,
  title={A modeling and computational study of the frustration index in signed networks},
  author={Aref, Samin and Mason, Andrew J and Wilson, Mark C},
  journal={Networks},
  volume={75},
  number={1},
  pages={95--110},
  year={2020},
  publisher={Wiley Online Library}
}

@inproceedings{chen2024scalable,
  title={Scalable Algorithm for Finding Balanced Subgraphs with Tolerance in Signed Networks},
  author={Chen, Jingbang and Mang, Qiuyang and Zhou, Hangrui and Peng, Richard and Gao, Yu and Ma, Chenhao},
  booktitle={Proceedings of the 30th ACM SIGKDD Conference on Knowledge Discovery and Data Mining},
  pages={278--287},
  year={2024}
}

@article{fortunato2010community,
  title={Community detection in graphs},
  author={Fortunato, Santo},
  journal={Physics reports},
  volume={486},
  number={3-5},
  pages={75--174},
  year={2010},
  publisher={Elsevier}
}

@article{harary1959measurement,
  title={On the measurement of structural balance},
  author={Harary, Frank},
  journal={Behavioral Science},
  volume={4},
  number={4},
  pages={316--323},
  year={1959},
  publisher={Wiley Online Library}
}

@article{figueiredo2014maximum,
  title={The maximum balanced subgraph of a signed graph: Applications and solution approaches},
  author={Figueiredo, Rosa and Frota, Yuri},
  journal={European Journal of Operational Research},
  volume={236},
  number={2},
  pages={473--487},
  year={2014},
  publisher={Elsevier}
}

@inproceedings{ordozgoiti2020finding,
  title={Finding large balanced subgraphs in signed networks},
  author={Ordozgoiti, Bruno and Matakos, Antonis and Gionis, Aristides},
  booktitle={Proceedings of The Web Conference 2020},
  pages={1378--1388},
  year={2020}
}

@inproceedings{mandaglio2020and,
  title={In and out: optimizing overall interaction in probabilistic graphs under clustering constraints},
  author={Mandaglio, Domenico and Tagarelli, Andrea and Gullo, Francesco},
  booktitle={Proceedings of the 26th ACM SIGKDD International Conference on Knowledge Discovery \& Data Mining},
  pages={1371--1381},
  year={2020}
}

@inproceedings{zha2001bipartite,
  title={Bipartite graph partitioning and data clustering},
  author={Zha, Hongyuan and He, Xiaofeng and Ding, Chris and Simon, Horst and Gu, Ming},
  booktitle={Proceedings of the tenth international conference on Information and knowledge management},
  pages={25--32},
  year={2001}
}

@article{peter2019encyclopedia,
  title={Encyclopedia of bioinformatics and computational biology},
  author={Peter, Swathik Clarancia and Dhanjal, Jaspreet Kaur and Malik, Vidhi and Radhakrishnan, Navaneethan and Jayakanthan, Mannu and Sundar, Durai and Sundar, Durai and Jayakanthan, Mannu},
  journal={Editor S. Ranganathan, M. Grib-skov, K. Nakai, and C. Sch{\"o}nbach},
  pages={661--676},
  year={2019}
}

@inproceedings{adamic2005political,
  title={The political blogosphere and the 2004 US election: divided they blog},
  author={Adamic, Lada A and Glance, Natalie},
  booktitle={Proceedings of the 3rd international workshop on Link discovery},
  pages={36--43},
  year={2005}
}

@article{yang2007community,
  title={Community mining from signed social networks},
  author={Yang, Bo and Cheung, William and Liu, Jiming},
  journal={IEEE transactions on knowledge and data engineering},
  volume={19},
  number={10},
  pages={1333--1348},
  year={2007},
  publisher={IEEE}
}

@article{doreian1996partitioning,
  title={A partitioning approach to structural balance},
  author={Doreian, Patrick and Mrvar, Andrej},
  journal={Social networks},
  volume={18},
  number={2},
  pages={149--168},
  year={1996},
  publisher={Elsevier}
}

@book{van2000asymptotic,
	title={Asymptotic Statistics},
	author={Van der Vaart, Aad W.},
	year={2000},
	publisher={Cambridge University Press},
	series={Cambridge Series in Statistical and Probabilistic Mathematics}
}

@article{giscard2017evaluating,
  title={Evaluating balance on social networks from their simple cycles},
  author={Giscard, Pierre-Louis and Rochet, Paul and Wilson, Richard C},
  journal={Journal of Complex Networks},
  volume={5},
  number={5},
  pages={750--775},
  year={2017},
  publisher={Oxford University Press}
}

@article{facchetti2011computing,
  title={Computing global structural balance in large-scale signed social networks},
  author={Facchetti, Giuseppe and Iacono, Giovanni and Altafini, Claudio},
  journal={Proceedings of the National Academy of Sciences},
  volume={108},
  number={52},
  pages={20953--20958},
  year={2011},
  publisher={National Academy of Sciences}
}

@article{xiong2020logic,
  title={On the logic of balance in social networks},
  author={Xiong, Zuojun and {\AA}gotnes, Thomas},
  journal={Journal of Logic, Language and Information},
  volume={29},
  number={1},
  pages={53--75},
  year={2020},
  publisher={Springer}
}

@article{aref2020detecting,
  title={Detecting coalitions by optimally partitioning signed networks of political collaboration},
  author={Aref, Samin and Neal, Zachary},
  journal={Scientific reports},
  volume={10},
  number={1},
  pages={1506},
  year={2020},
  publisher={Nature Publishing Group UK London}
}

@inproceedings{huang2022pole,
  title={Pole: Polarized embedding for signed networks},
  author={Huang, Zexi and Silva, Arlei and Singh, Ambuj},
  booktitle={Proceedings of the Fifteenth ACM International Conference on Web Search and Data Mining},
  pages={390--400},
  year={2022}
}

@article{capozzi2023analyzing,
  title={Analyzing and visualizing polarization and balance with signed networks: the US congress case study},
  author={Capozzi, Arthur and Semeraro, Alfonso and Ruffo, Giancarlo},
  journal={Journal of Complex Networks},
  volume={11},
  number={4},
  pages={cnad027},
  year={2023},
  publisher={Oxford University Press}
}

@article{liu2018multi,
  title={Multi-domain networks association for biological data using block signed graph clustering},
  author={Liu, Ye and Ng, Michael K and Wu, Stephen},
  journal={IEEE/ACM Transactions on Computational Biology and Bioinformatics},
  volume={17},
  number={2},
  pages={435--448},
  year={2018},
  publisher={IEEE}
}

@article{yang2020graph,
  title={Graph-based prediction of protein-protein interactions with attributed signed graph embedding},
  author={Yang, Fang and Fan, Kunjie and Song, Dandan and Lin, Huakang},
  journal={BMC bioinformatics},
  volume={21},
  number={1},
  pages={323},
  year={2020},
  publisher={Springer}
}

@inproceedings{zhang2012signed,
  title={Signed network propagation for detecting differential gene expressions and DNA copy number variations},
  author={Zhang, Wei and Johnson, Nicholas and Wu, Baolin and Kuang, Rui},
  booktitle={Proceedings of the ACM conference on bioinformatics, computational biology and biomedicine},
  pages={337--344},
  year={2012}
}

@article{ehsani2020structure,
  title={The structure of stock markets as signed networks},
  author={Ehsani, Maryam},
  journal={Journal of Industrial and Systems Engineering},
  year={2020}
}

@article{hashemi2022review,
  title={The review of ecological network indicators in graph theory context: 2014--2021},
  author={Hashemi, Rastegar and Darabi, Hassan},
  journal={International Journal of Environmental Research},
  volume={16},
  number={2},
  pages={24},
  year={2022},
  publisher={Springer}
}

@article{xiao2017mapping,
  title={Mapping the ecological networks of microbial communities},
  author={Xiao, Yandong and Angulo, Marco Tulio and Friedman, Jonathan and Waldor, Matthew K and Weiss, Scott T and Liu, Yang-Yu},
  journal={Nature communications},
  volume={8},
  number={1},
  pages={2042},
  year={2017},
  publisher={Nature Publishing Group UK London}
}

@article{boginski2005statistical,
  title={Statistical analysis of financial networks},
  author={Boginski, Vladimir and Butenko, Sergiy and Pardalos, Panos M},
  journal={Computational statistics \& data analysis},
  volume={48},
  number={2},
  pages={431--443},
  year={2005},
  publisher={Elsevier}
}

@incollection{abelson2017symbolic,
  title={Symbolic psycho-logic: A model of attitudinal cognition},
  author={Abelson, Robert P and Rosenberg, Milton J},
  booktitle={Attitude change},
  pages={86--115},
  year={2017},
  publisher={Routledge}
}

@article{hartmanis1982computers,
  title={Computers and intractability: a guide to the theory of np-completeness (michael r. garey and david s. johnson)},
  author={Hartmanis, Juris},
  journal={Siam Review},
  volume={24},
  number={1},
  pages={90},
  year={1982},
  publisher={Society for Industrial and Applied Mathematics}
}

@article{tarjan1972depth,
  title={Depth-first search and linear graph algorithms},
  author={Tarjan, Robert},
  journal={SIAM journal on computing},
  volume={1},
  number={2},
  pages={146--160},
  year={1972},
  publisher={SIAM}
}
